\newif\ifreport\reporttrue
\documentclass[journal]{IEEEtran}
\usepackage{setspace}
\usepackage{cite}
\usepackage{graphicx}
\usepackage{bm}
\usepackage[cmex10]{amsmath}
\usepackage{mathtools}
\usepackage{amsthm}
\usepackage{amsfonts}
\usepackage{amssymb}
\usepackage{hyperref}
\usepackage{algpseudocode}
\usepackage[lined, boxed, linesnumbered, ruled]{algorithm2e}
\usepackage{fixltx2e}%fixes latex float algorithm in 2 columns
\usepackage{tikz}
\usetikzlibrary{decorations.pathreplacing}
\theoremstyle{definition}
\newtheorem{definition}{Definition}
\newcommand{\age}{\Delta}
\usepackage{color}

\def\blue{\color{black}}
\def\red{\color{red}}
\newcommand{\ignore}[1]{}

\newtheorem{lemma}{Lemma}

\newtheorem{theorem}{Theorem}
\newtheorem{corollary}{Corollary}
\begin{document}
\IEEEoverridecommandlockouts
%\title{Maximizing Data Freshness for Information Updates}
%\title{Minimizing the Age of Information in Multi-Source Networks}
\title{Minimizing the Age of the Information to Multiple Sources}
\title{\huge	 Age-Optimal Updates of Multiple Information Flows}
\title{\huge	 Age-Optimal Updates of Multiple Information Flows: A Sample-path Approach}
\title{Age of Information Minimization for Multiple Information Flows: A Sample-path Approach}
\title{Age-Optimal Multi-Flow Status Updating with Errors: A Sample-Path Approach}

%\title{Multi-flow Age Minimization through Multiple Channels with Errors: A Sample-path Approach}

\author{Yin Sun and Sastry Kompella %, Elif Uysal, 
\thanks{This paper was presented in part at the IEEE INFOCOM Age of Information (AoI) Workshop in 2018 \cite{SunAoIWorkshop2018}.}
\thanks{Yin Sun's work is supported in part by the NSF grant CNS-2239677 and the ARO grant W911NF-21-1-0244.}
\thanks{Yin Sun is with the Department of ECE, Auburn University, Auburn, AL 36849 USA, email: yzs0078@auburn.edu.}
%\thanks{Elif Uysal is with the Department of EEE, Middle East Technical University, Ankara, Turkey, email: uelif@metu.edu.tr.}
\thanks{Sastry Kompella is with Nexcepta Inc., Gaithersburg, MD 20878 USA, e-mail: sk@ieee.org.}

}

\maketitle
\thispagestyle{plain}
\pagestyle{plain}
% !TEX root = ./sampling_BM.tex
\begin{abstract}
%In this paper, we study an age of information minimization problem, where multiple flows of update packets are sent over multiple servers to their destinations. Two online scheduling policies are proposed. When  the packet generation and arrival times  are synchronized across the flows, the proposed policies are shown to be (near) optimal for minimizing any \emph{time-dependent}, \emph{symmetric}, and \emph{non-decreasing} penalty function of the ages of the flows over time in a stochastic ordering sense. {\red revise abstract later}

In this paper, we study an age of information minimization problem in \emph{continuous-time} and \emph{discrete-time} status updating systems that involve \emph{multiple packet flows}, \emph{multiple servers}, and \emph{transmission errors}. Four scheduling policies are proposed. 
We develop a unifying sample-path approach and use it to show that, when the packet generation and arrival times  are synchronized across the flows, the proposed policies are (near) optimal for minimizing any \emph{time-dependent}, \emph{symmetric}, and \emph{non-decreasing} penalty function of the ages of the flows over time in a stochastic ordering sense.
\end{abstract}

\section{Introduction}
%{\red ideas: where do we have strong need for multi-server multi-flow scheduling? IoT, social networks, online gaming, news, and notifications.} 
%In computer and communcation networks, 

%The  increased availability of network connected mobile devices has spurred a plethora of industrial and daily life applications involving real-time remote measurement, tracking, and control, which relies heavily on the availability of fresh information updates. 
%
%These applications, whether their end user is a person (e.g., a social network or news feed, driving directions) or device (e.g., industrial environmental monitoring, vehicle sensor status, automated driving), are characterized by a dependence on status updates, that is, information packets that contain recently sampled data. Status updates are desired to be sufficiently \emph{fresh}, or \emph{timely} for the application at hand. 

In many information-update and networked control systems, such as news updates, stock trading, autonomous driving, remote surgery, robotics control, and real-time surveillance, information usually has the greatest value when it is fresh. A metric for information freshness, called  \emph{age of information} or simply  \emph{age}, was introduced in \cite{Song1990,KaulYatesGruteser-Infocom2012}. Consider a flow of status update packets that are sent from a source to a destination through a channel. 
Let $U(t)$ be the time stamp (i.e., generation time) of the {newest update that the destination has received} by time $t$. {Age of information, as a function of time $t$, is defined as} 
$\Delta (t) = t - U(t)$, which is the time elapsed since the newest update was generated.

%how to reduce the age $\Delta (t)$ and keep the information fresh, 

In recent years, there have been a lot of research efforts on (i) analyzing the distributional  quantities of age $\Delta (t)$ for various network models and (ii) controlling $\Delta (t)$ to keep the destination's information as fresh as possible, e.g., 
\cite{SunAoIWorkshop2018,Song1990,KaulYatesGruteser-Infocom2012,2012CISS-KaulYatesGruteser,2012ISIT-YatesKaul,LiInfocom2015,6875100,CostaCodreanuEphremides_TIT, KamKompellaEphremidesTIT,Icc2015Pappas,2015ISITHuangModiano,Suninfocom2016,AgeOfInfo2016,Bedewy2016,BedewyJournal2017,Bedewy2017,BedewyMultihop2017,SunBook,Kosta2017Age,Yates2016, AliTCOM2022,IgorAllerton2016,HsuTWC2017,Vishrant2017,Arunabh2019,He2018,8822722,8812616,9137714,8406891,SunMutualInformation2018,SunNonlinear2019,shisher2021age,ShisherMobiHoc22, shisher2023learning0, shisher2023learning, pan2022age, pan2022optimal, bedewy2021low, ornee2021sampling, bedewy2021optimal, tang2022sampling, ornee2023whittle,yates2021AgeSurvey}. If there is  a single flow of status update packets, the Last Generated, First Served (LGFS) update transmission policy, in which the last generated packet is served the first, has been shown to be (nearly) optimal for minimizing the age process $\{\age(t),t\geq 0\}$ in a stochastic ordering sense for queueing networks with multiple servers or multiple hops \cite{Bedewy2016,BedewyJournal2017,Bedewy2017,BedewyMultihop2017,SunBook}. 
These results hold for arbitrary packet generation times at the information source (e.g., a sensor) and arbitrary packet arrival times at the transmitter's queueing buffer; they also hold for minimizing any non-decreasing functional %\footnote{A functional is a mapping from functions to real numbers.} 
$\phi(\{\age(t),t\geq 0\})$ of the age process $\{\age(t),t\geq 0\}$. If packets arrive at the queue in the order of their generation times, then the LGFS policy reduces to the Last Come, First Served (LCFS) policy, thus demonstrating the (near) age-optimality of the LCFS policy. These studies motivated us to delve deeper into the design of scheduling policies to minimize age of information in more complex networks involving \emph{multiple flows of status update packets} and \emph{transmission errors}, where each flow is from one source node to a destination node. In this scenario, the transmission scheduler must compare not only packets from the same flow, but also packets from different flows. Additionally, the presence of transmission errors adds an additional layer of complexity to the scheduling problem. As a result, addressing these challenges becomes crucial in achieving efficient age minimization in such systems. 

\begin{figure}
\centering 
\includegraphics[width=0.49\textwidth]{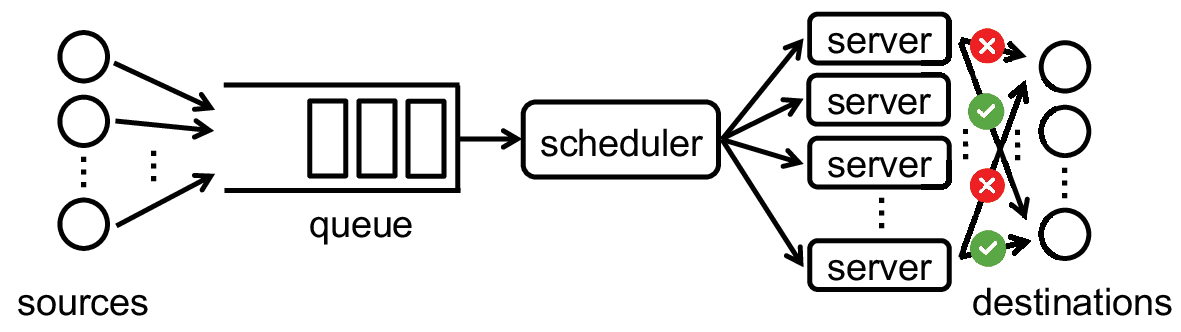} 
\vspace{-0mm}
\caption{System model.}
% work--efficiency ordering holds for any priorities of the jobs.
\label{fig_model} 
\vspace{-5mm}
\end{figure} 

In this paper, we investigate age-optimal scheduling in \emph{continuous-time} and \emph{discrete-time} status updating systems that involve \emph{multiple flows}, \emph{multiple servers}, and \emph{transmission errors}, as illustrated in Figure \ref{fig_model}. Each server can transmit packets to their respective destinations, one packet at a time. Different servers are not allowed to simultaneously transmit packets from the same flow. {We assume that the packet generation and arrival times are \emph{synchronized} across the flows. 
In other words, when a packet from flow $n$ arrives at the queue at time $A_i$, with its generation time denoted as $S_i$ (where $S_i\leq A_i$), one corresponding packet from each flow simultaneously received at time $A_i$, and all of these packets were generated at the same time $S_i$.
%This assumption is a generalized version of  the model in \cite{IgorAllerton2016}.
In practice, synchronized packet generations and arrivals occur when there is a single source and multiple destinations (e.g.,  \cite{IgorAllerton2016}), or in periodic sampling where multiple sources are synchronized by the same clock, which is common in  monitoring and control systems \ifreport
(e.g.,  \cite{Phadke1994,Sivrikaya2004})\fi.}
We develop a unifying sample-path approach and use it to show that the proposed scheduling policies can achieve optimal or near-optimal age performance in a quite strong sense (i.e., in terms of stochastic ordering of age-penalty stochastic processes). 
The contributions of this paper are summarized as follows:
\begin{itemize}
\item Let $\bm{\Delta}(t)$ denote the age vector of multiple flows. We introduce an age penalty function $p_t (\bm\age(t))$ to represent the level of dissatisfaction for having aged information at the destinations at time $t$, where $p_t$ can be any \emph{time-dependent}, \emph{symmetric}, and \emph{non-decreasing} function of the age vector $\bm{\Delta}(t)$. 

\item For continuous-time status updating systems with one or multiple flows, one or multiple servers, and \emph{i.i.d.}~exponential transmission times, we propose a \emph{Preemptive, Maximum Age First, Last Generated First Served (P-MAF-LGFS) scheduling policy}.\footnote{
%We note that this P-MAF-LGFS policy is more general than that presented in \cite{SunAoIWorkshop2018} with the same name, which only applies to single-server systems. the P-MAF-LGFS policy discussed here is applicable to both single-server and multi-server systems. It broadens the scope of the original P-MAF-LGFS policy designed for single-server setups, as introduced in \cite{SunAoIWorkshop2018}, to encompass the more general multi-server scenario.
This new P-MAF-LGFS policy is suitable for both single-server and multi-server systems, whereas the original P-MAF-LGFS policy, as presented in \cite{SunAoIWorkshop2018}, was specifically tailored for single-server scenarios. 
} 
If the packet generation and arrival times are synchronized across the flows, then for any age penalty function $p_t$ defined above, any number of flows, any number of servers, any synchronized packet generation and arrival times, and regardless the presence of transmission errors or not, the P-MAF-LGFS policy is proven to minimize the continuous-time age penalty process $\{p_t (\bm\age(t)), t\geq 0\}$ among all causal policies in a stochastic ordering sense (see Theorem \ref{thm1} and Corollary \ref{coro1}). Theorem \ref{thm1} is more general than \cite[Theorem 1]{SunAoIWorkshop2018}, as the latter was established for the special case of single-server status updating systems without transmission errors. In addition, if packet replication is allowed, we show that a \emph{Preemptive, Maximum Age First, Last Generated First Served scheduling policy with packet Replications (P-MAF-LGFS-R)} is age-optimal for minimizing the age penalty process $\{p_t (\bm\age(t)), t\geq 0\}$ in terms of stochastic ordering (see Corollary \ref{corollary_new}).

\item For continuous-time status updating systems with one or multiple flows, one or multiple servers, and \emph{i.i.d.}~New-Better-than-Used (NBU) transmission times (which include exponential transmission times as a special case),
 age-optimal multi-flow scheduling is quite difficult to achieve. In this case, 
we consider an age lower bound called the \emph{Age of Served Information} and propose a \emph{Non-Preemptive, Maximum Age of Served Information First, Last Generated First Served (NP-MASIF-LGFS) scheduling policy}. The NP-MASIF-LGFS policy is shown to be near age-optimal. Specifically, it is within an additive gap from the optimum for minimizing the expected time-average of the average age of the flows, where the gap is equal to the mean transmission time of one packet (see Theorem \ref{thm3} and Corollary \ref{coro4}). This additive sub-optimality gap is quite small. 
%Numerical evaluations are provided to verify our (near) age optimality results. %{\blue Some possible extensions are discussed at the end of the paper.}

\item For discrete-time status updating systems with one or multiple flows and one or multiple servers, we propose a \emph{Discrete Time, Maximum Age First, Last Generated First Served (DT-MASIF-LGFS) scheduling policy}. If the packet generation and arrival times are synchronized across the flows, then for any age penalty function $p_t$, any number of flows, any number of servers, any synchronized packet generation and arrival times, and regardless the presence of transmission errors or not, the DT-MAF-LGFS policy is proven to minimize the discrete-time age penalty process $\{p_t (\bm\age(t)), t= 0, T_s, 2 T_s, \ldots\}$ among all causal policies in a stochastic ordering sense, where $T_s$ is the fundamental time unit of the discrete-time systems (see Theorem \ref{thm4}). %Theorem \ref{thm4} is more general than \cite[Theorem 1]{IgorAllerton2016}, because the former allows for more general packet generation and arrival times, and a broader range of age penalty functions.

\end{itemize}

%A comparison with related work is presented in Section \ref{sec_related_work}. 
\ifreport
Our results can be potentially applied to: (i) cloud-hosted Web services where the servers in Figure \ref{fig_model} represent a pool of threads (each for a TCP connection) connecting a front-end proxy node to clients \cite{Fox:1997:CSN:269005.266662}, (ii) industrial robotics and factory automation systems where multiple sensor-output flows are sent to a wireless AP and then forwarded to a system monitor and/or controller \cite{Gungor2009}, and (iii) Multi-access Edge Computing (MEC) that can process fresh data (e.g., data for video analytics, location services, and IoT) locally at the very edge of the mobile network. % \cite{MEC}. 
\fi

% !TEX root = ./Age_of_Info_multi_source.tex

\section{Related Work}\label{sec_related_work}

The age of information concept has attracted a significant surge of research interest; see, e.g., \cite{SunAoIWorkshop2018,Song1990,KaulYatesGruteser-Infocom2012,2012CISS-KaulYatesGruteser,2012ISIT-YatesKaul,LiInfocom2015,6875100,CostaCodreanuEphremides_TIT, KamKompellaEphremidesTIT,Icc2015Pappas,2015ISITHuangModiano,Suninfocom2016,AgeOfInfo2016,Bedewy2016,BedewyJournal2017,Bedewy2017,BedewyMultihop2017,SunBook,Yates2016, AliTCOM2022,IgorAllerton2016,HsuTWC2017,Vishrant2017,Arunabh2019,He2018,8822722,8812616,9137714,8406891,SunMutualInformation2018,SunNonlinear2019,shisher2021age,ShisherMobiHoc22, shisher2023learning0, shisher2023learning, pan2022age, pan2022optimal, bedewy2021low, ornee2021sampling, bedewy2021optimal, tang2022sampling, ornee2023whittle, Kosta2017Age} and a recent survey \cite{yates2021AgeSurvey}. 
Initially, research efforts were centered on analyzing and comparing the age performance of different queueing disciplines, such as First-Come, First-Served (FCFS) \cite{KaulYatesGruteser-Infocom2012,2012ISIT-YatesKaul,KamKompellaEphremidesTIT,2015ISITHuangModiano}, preemptive and non-preemptive Last-Come, First-Served (LCFS) \cite{2012CISS-KaulYatesGruteser,Yates2016}, and packet management \cite{CostaCodreanuEphremides_TIT, Icc2015Pappas}. In \cite{Bedewy2016,BedewyJournal2017,Bedewy2017,BedewyMultihop2017,SunBook}, a sample-path approach was developed to prove that Last-Generated, First-Served (LGFS)-type policies are optimal or near-optimal for minimizing a broad class of age metrics in multi-server and multi-hop queueing networks with a single packet flow. When packets arrive in the order of their generation times, the LGFS policy becomes the well-known Last Come, First Served (LCFS) policy. Hence, the LCFS policy is (near) age-optimal in these queueing networks. 

In recent years, researchers have expanded the aforementioned studies to consider age minimization in multi-flow discrete-time status updating systems  \cite{IgorAllerton2016,HsuTWC2017,Vishrant2017,Arunabh2019}. In \cite{IgorAllerton2016}, the authors utilized a sample-path method to establish the optimality of the Maximum Age First (MAF) policy in minimizing the time-averaged sum age of multiple flows. This investigation focused on discrete-time systems with periodic arrivals and a single broadcast channel, which is susceptible to \emph{i.i.d.} transmission errors.
Moreover, in \cite{HsuTWC2017}, a Markov decision process (MDP) approach was adopted to prove that the MAF policy minimizes the time-averaged sum age of multiple flows in discrete-time systems with Bernoulli arrivals, a single broadcast channel, and no buffer. In this bufferless setup, arriving packets are discarded if they cannot be transmitted immediately in the arriving time slot. In \cite{Vishrant2017}, the authors studied discrete-time systems with multiple flows and multiple ON/OFF channels, where the state of each channel (ON/OFF) is known for making scheduling decisions. It was demonstrated that a Max-Age Matching policy is asymptotically optimal for minimizing non-decreasing symmetric functions of the age of the flows as the numbers of flows and channels increase. In \cite{Arunabh2019}, it was shown that the MAF policy minimizes the Maximum Age of multiple flows in discrete-time systems with periodic arrivals and a single broadcast channel susceptible to \emph{i.i.d.} transmission errors, where the transmission error probability may vary across the flows. In \cite{Ruogu2013}, a sample-path method was employed to demonstrate that the round-robin policy minimizes a service regularity metric called \emph{time-since-last-service} in discrete-time systems with multiple flows and transmission errors. In the definition of time-since-last-service, a user can receive service even if its queue is empty. Consequently, time-since-last-service bears similarities to the age of information concept, albeit these two metrics are different. 
The present paper, alongside its conference version \cite{SunAoIWorkshop2018}, complements the aforementioned studies in several essential ways: (i) It considers general time-dependent, symmetric, and non-decreasing age penalty functions $p_t$. (ii) Both continuous-time and discrete-time systems with multiple flows, multiple channels (a.k.a. servers), and transmission errors are investigated. (iii) The paper establishes near age-optimal scheduling results in scenarios where achieving age-optimality is inherently challenging.

%In \cite{IgorAllerton2016}, the expected time-average of the weighted sum age of multiple sources was minimized in a broadcast network with an ON-OFF channel and periodic arrivals, where only one source is scheduled at a time and the scheduler does not know the current ON-OFF channel state. When the network is symmetric and the weights are equal, a sample-path method was used to show that the maximum age first (MAF)  policy is optimal. Further, a sub-optimal Whittle's index method was used to handle the general asymmetric cases. In \cite{HsuTWC2017}, for symmetric Bernoulli arrivals and an always-ON channel with no buffers, the MAF policy was shown to be optimal for minimizing  the expected time-average of the sum age of multiple sources. In addition, Markov decision process (MDP) methods were used to handle the general scenarios with asymmetric arrivals and a buffer, where the optimal policies are shown to be switch-type. 

\section{System Model}\label{sec:model}
\subsection{Notations and Definitions}\label{sec:def}
% %, with $|\mathcal{S}|$ denoting the cardinality of $\mathcal{S}$.
%%For any random variable ${X}$ and any event $\mathcal{A}$, let $[{X}|\mathcal{A}]$ denote a random variable with the conditional distribution of ${X}$ for given $\mathcal{A}$. 

%For any random variable $Z$ and event $\mathcal{A}$, let $[Z|\mathcal{A}]$ denote a random variable with the conditional distribution of $Z$ for given $\mathcal{A}$.
% and $\mathbb{E}[Z|\mathcal{A}]$ denote the conditional expectation of $Z$ for given $\mathcal{A}$. 
%Let $u$ and $1_A$ denote the unit step function and indicator function of event $A$, respectively, i.e.,
%\begin{align}
%u(t) = \left\{\begin{array}{l l} 1,& \text{if}~t\geq0;\\0,& \text{if}~t<0,\end{array}\right.~~~1_A(x) = \left\{\begin{array}{l l} 1,& \text{if}~x\in A;\\0,& \text{if}~x\notin A.\end{array}\right.\nonumber
%\end{align}

%Let $\bm{x} = (x_1, x_2,\ldots,$ $x_m)$ and $\bm{y} = (y_1, y_2,\ldots,y_m)$ be two vectors in $\mathbb{R}^m$, then we denote $\bm{x} \leq \bm{y}$ if $x_i \leq y_i$ for $i = 1,2,\ldots,m$. A set $U \subseteq \mathbb{R}^m$ is called \emph{upper}, if for all $\bm{x} \in U$ and $\bm{y}\geq \bm{x}$ it holds that $\bm{y} \in U$. 

We  use lower case letters such as $x$ and $\bm{x}$, respectively, to represent deterministic scalars and vectors. In the vector case, a subscript will index the components of a vector, such as $x_i$.
%We use $x_{[i]}$ %and $x_{(i)}$, respectively, 
%to denote the $i$-th largest %and the $i$-th smallest 
%component of $\bm{x}$. 
We use $x_{[i]}$ %and $x_{(i)}$, respectively, 
to denote the $i$-th largest %and the $i$-th smallest 
component of vector $\bm{x}$. Let $\bm 0$ denote   a vector with all 0 components.
%Let $\bm{x}_{\uparrow}=(x_{(1)},\ldots,x_{(n)})$ %and $\bm{x}_{\downarrow}=(x_{[1]},\ldots,x_{[n]})$, respectively,
%denote the increasing %and decreasing
%rearrangements of $\bm{x}$. 
A function $f: \mathbb{R}^n\rightarrow \mathbb{R}$ is termed \emph{symmetric} if $f(\bm{x})= f(x_{[1]},\ldots, x_{[n]})$ for all $\bm{x} \in \mathbb{R}^n$. A function $f: \mathbb{R}^n\rightarrow \mathbb{R}$ is termed \emph{separable} if there exists functions $f_1,\ldots,f_n$ of one variable such that $f(\bm{x}) = \sum_{i=1}^n f_i(x_i)$ for all $\bm{x} \in \mathbb{R}^n$. 
The composition of functions $f$ and $g$ is denoted by $f \circ g( x) = f(g ( x))$. 
 For any $n$-dimensional vectors $\bm{x}$ and $\bm{y}$, the elementwise vector ordering $x_i\leq y_i$, $i=1,\ldots,n$, is denoted by $\bm{x} \leq \bm{y}$. 
 %Further, $\bm{x}$ is said to be \emph{majorized} by $\bm{y}$, denoted by $\bm{x}\prec\bm{y}$, if (i) $\sum_{i=1}^j x_{[i]} \leq \sum_{i=1}^j y_{[i]}$, $j=1,\ldots,n-1$ and (ii) $\sum_{i=1}^n x_{[i]} = \sum_{i=1}^n y_{[i]}$ \cite{Marshall2011}. In addition, $\bm{x}$ is said to be  \emph{weakly majorized by $\bm{y}$ from below}, denoted by $\bm{x}\prec_{\text{w}}\bm{y}$, if $\sum_{i=1}^j x_{[i]} \leq \sum_{i=1}^j y_{[i]}$, $j=1,\ldots,n$; $\bm{x}$ is said to be  \emph{weakly majorized by $\bm{y}$ from above}, denoted by $\bm{x}\prec^{\text{w}}\bm{y}$, if $\sum_{i=1}^j x_{(i)} \geq \sum_{i=1}^j y_{(i)}$, $j=1,\ldots,n$ \cite{Marshall2011}.
%A function that preserves the majorization order is called a Schur convex function. Specifically, $f: \mathbb{R}^n\rightarrow \mathbb{R}$ is termed \emph{Schur convex} if $f(\bm{x})\leq f(\bm{y})$ for all $\bm{x}\prec\bm{y}$ \cite{Marshall2011}.
Let $\mathcal{A}$ 
and $\mathcal{U}$ 
denote sets and events. 
For all random variable ${X}$ and event $\mathcal{A}$, let $[{X}|\mathcal{A}]$ denote a random variable with the conditional distribution of ${X}$ for given $\mathcal{A}$. We will need the following definitions:

%\begin{definition}\emph{Majorization \cite{Marshall2011}}: For any $n$-dimensional vectors $\bm{x}$ and $\bm{y}$, $\bm{x}$ is said to be \emph{majorized} by $\bm{y}$, denoted by $\bm{x}\prec\bm{y}$, if (i) $\sum_{i=1}^j x_{[i]} \leq \sum_{i=1}^j y_{[i]}$, $j=1,\ldots,n-1$ and (ii) $\sum_{i=1}^n x_{[i]} = \sum_{i=1}^n y_{[i]}$. In addition, $\bm{x}$ is said to be  \emph{weakly majorized by $\bm{y}$ from below}, denoted by $\bm{x}\prec_{\text{w}}\bm{y}$, if $\sum_{i=1}^j x_{[i]} \leq \sum_{i=1}^j y_{[i]}$, $j=1,\ldots,n$.
%% $\bm{x}$ is said to be  \emph{weakly majorized by $\bm{y}$ from above}, denoted by $\bm{x}\prec^{\text{w}}\bm{y}$, if $\sum_{i=1}^j x_{(i)} \geq \sum_{i=1}^j y_{(i)}$, $j=1,\ldots,n$ . 
%\end{definition}
%
%\begin{definition}\emph{Schur Convexity \cite{Marshall2011}}:  A function that preserves the majorization order is called a Schur convex function. Specifically, $f: \mathbb{R}^n\rightarrow \mathbb{R}$ is termed \emph{Schur convex} if $f(\bm{x})\leq f(\bm{y})$ for all $\bm{x}\prec\bm{y}$ \cite{Marshall2011}. 
%\end{definition}
% Define $x\wedge y=\min\{x,y\}$.

\begin{definition}\label{def_variable}
\emph{Stochastic Ordering of Random Variables \cite{StochasticOrderBook}}: 
A random variable ${X}$ is said to be \emph{stochastically smaller} than another random variable ${Y}$, denoted by ${X}\leq_{\text{st}}{Y}$, if 
\begin{align}
\Pr({X}>t) \leq \Pr({Y}>t),~\forall~t\in \mathbb{R}. 
\end{align}
\end{definition}
\begin{definition}\label{def_vector}
\emph{Stochastic Ordering of Random Vectors \cite{StochasticOrderBook}}: 
A set $\mathcal{U} \subseteq \mathbb{R}^n$ is called \emph{upper}, if $\bm{y} \in \mathcal{U}$ whenever $\bm{y}\geq \bm{x}$ and $\bm{x} \in \mathcal{U}$. 
Let $\bm{X}$ and $\bm{Y}$ be two $n$-dimensional random vectors, $\bm{X}$ is said to be \emph{stochastically smaller} than $\bm{Y}$, denoted by $\bm{X}\leq_{\text{st}}\bm{Y}$, if 
\begin{align}
\Pr(\bm{X}\in \mathcal{U}) \leq \Pr(\bm{Y}\in \mathcal{U})~\text{for all upper sets}~\mathcal{U}\subseteq \mathbb{R}^n.
\end{align}
\end{definition}
\begin{definition}\label{def_process}
\emph{Stochastic Ordering of Stochastic Processes \cite{StochasticOrderBook}}: 
Let $\{X(t),t\in [0,\infty) \}$ and $\{Y(t),t\in [0,\infty) \}$ be two stochastic processes, $\{X(t), t\in [0,\infty) \}$ is said to be \emph{stochastically smaller} than $\{Y(t),t\in [0,\infty) \}$, denoted by $\{X(t),t\in [0,\infty) \}\leq_{\text{st}}\{Y(t),t\in [0,\infty)\}$, if for all integer $n$ and $0\leq t_1< t_2<\ldots<t_n$, it holds that 
\begin{align}
\!\!\!(X(t_1),X(t_2),\ldots,X(t_n)) \!\leq_{\text{st}}\! (Y(t_1),Y(t_2),\ldots,Y(t_n)).\!\!\!
\end{align}
\end{definition}

A functional is a mapping from functions to real numbers. A functional $\phi$ is termed \emph{non-decreasing} if $\phi(\{X(t), t\in [0,\infty)\}) \leq \phi(\{Y(t), t\in [0,\infty)\})$ whenever $X(t)\leq Y(t)$ for $t\in [0,\infty)$.
We remark that $\{X(t),t\in [0,\infty) \}\leq_{\text{st}}\{Y(t),t\in [0,\infty)\}$  if, and only if,  \cite{StochasticOrderBook}
\begin{equation}\label{eq_order}
\mathbb{E}[\phi(\{X(t), t\in [0,\infty)\} )] \leq \mathbb{E}[\phi(\{Y(t), t\in [0,\infty)\} )]
\end{equation}
holds for all non-decreasing functional $\phi$, provided that the expectations in \eqref{eq_order} exist.

\subsection{Queueing System Model}\label{sec:queuemodel}
%{\red Can be generalized to multiple servers.}

Consider the status updating system illustrated in Fig. \ref{fig_model}, where $N$ flows of status update packets are sent through a queue with an infinite buffer and $M$ servers. Let $s_n$ and $d_n$ denote the source and destination nodes of flow $n$, respectively. It is possible for multiple flows to share either the same source node or the same destination node.

A scheduler assigns packets from the transmitter's queue to servers over time. The queue contains packets from different flows, and each packet can be assigned to any available server. Each server is capable of transmitting only one packet at a time. Different servers are not allowed to simultaneously transmit packets from the same flow.
%Each of the servers can transmit any packet  to the associated destination, one packet at a time. %
The packet transmission times are independent and identically distributed (\emph{i.i.d.})~across both servers and packets, with a finite mean  $1/\mu$. The packet transmissions are susceptible to \emph{i.i.d.} errors with an error probability $q\in[0,1)$, 
occurring at the end of the packet transmission time intervals.
The scheduler is made aware of transmission errors once they occur. In the event of such a error, the packet is promptly returned to the queue, where it awaits the next transmission opportunity. if $q= 0$, then there is no transmission errors.

%Hence, the packet service time distribution depends on the server, rather than the packet. 

%In practice, the servers can be TCP/HTTP connections or wireless communication channels that can be assigned to different flows. 
% {\red and a buffer size $B\geq N$}
% infinite buffer size.\footnote{It is easy to check that the results in this paper hold if the buffer size $B$ is finite and $B\geq N$.} 

% a transmitter sends update packets to $R$ receivers through $M$ communication servers, where a server could be a wireless communication channel, a TCP/HTTP connection, etc. 
The system starts to operate at time $t=0$.  The $i$-th packet of flow $n$ is generated  at the source node $s_n$ at time $S_{n,i}$, arrives at the queue at time $A_{n,i}$, and is delivered to the destination $d_n$ at time $D_{n,i}$ such that $0\leq S_{n,1} \leq S_{n,2}\leq\ldots$ and $S_{n,i}\leq A_{n,i}\leq D_{n,i}$.\footnote{This paper allows $S_{n,i}\leq A_{n,i}$, which is  more general than the conventional assumption $S_{n,i}= A_{n,i}$  adopted in  related literature.} 
We consider the following class of \emph{synchronized} packet generation and arrival processes:
%We assume that the packet generation and arrival times are \emph{synchronized} across the $N$ flows, as defined below:

\begin{definition} \emph{Synchronized Packet Generations and Arrivals:}
The packet generation and arrival processes are said to be \emph{synchronized} across  the $N$ flows, if there exist two sequences $\{S_1, S_2,\ldots\}$ and $\{A_1,A_2,\ldots\}$ such that for all $i=1,2,\ldots,$ and $n=1,\ldots,N$
\begin{align}\label{eq_synchronized}
S_{n,i} = S_i,~A_{n,i} = A_i.
\end{align}
\end{definition}
We note that the sequences $\{S_1, S_2,\ldots\}$ and $\{A_1,A_2,\ldots\}$ in  \eqref{eq_synchronized} are \emph{arbitrary}. Hence,
\emph{out-of-order arrivals}, e.g., $S_i < S_{i+1}$ but $A_i > A_{i+1}$, are allowed. In the special case that the system has a single flow ($N=1$), the packet generation times $S_{n,1}$ and arrival times $A_{n,1}$ of this flow are {arbitrarily} given without any constraint. Age-optimal scheduling in this special case has been previously studied in \cite{Bedewy2016,BedewyJournal2017,Bedewy2017,BedewyMultihop2017}. 

%synchronized sampling and arrivals mean that the packet generation and arrival processes of this flow can be arbitrary; 

%More general arrival processes will be studied later.

Let $\pi$ represent a scheduling policy that determines how to assign  packets from the  queue to servers over time. Let $\Pi$ denote the set of all \emph{causal} scheduling policies in which the scheduling decisions are made based on the history and current states of the system.
A scheduling policy is said to be \emph{preemptive} if a busy server can stop the transmission of the current packet and start sending another packet at any time; the preempted packet is stored back to the queue, waiting to be sent at a later time.
A scheduling policy is said to be \emph{non-preemptive} if each server must complete the transmission of the current packet before initiating the service of another packet. A scheduling policy is said to be \emph{work-conserving} if all servers remain busy whenever the queue contains packets waiting to be processed.  We use $\Pi_{np}$ to denote the set of non-preemptive and causal scheduling policies, where $\Pi_{np}\subset \Pi$. 
%and use $\Pi_{npwc}$ to denote the set of non-preemptive, work-conserving, and causal policies, such that $\Pi_{npwc} \subset\Pi_{np}\subset \Pi$.
Let 
\begin{align}
\mathcal{I}=\{S_{i}, A_{i},~ i=1,2,\ldots\} 
\end{align}
denote the synchronized packet generation and arrival times of the flows. 
We assume that the packet generation/arrival times $\mathcal{I}$, the packet transmission times, and the transmission errors are
governed  by three \emph{mutually independent}  stochastic processes, none of which are influenced by the  scheduling policy.

\subsection{Age  Metrics}

Among the packets that have been delivered to the destination  $d_n$ of flow $n$ by time $t$, the freshest packet was generated at time  
\begin{align}
U_{n} (t) =\max_i \{S_{n,i}: D_{n,i} \leq t\}.
\end{align}
 \emph{Age of information}, or simply \emph{age}, for flow $n$ is defined as \cite{Song1990,KaulYatesGruteser-Infocom2012}
\begin{align}\label{eq_age}
\Delta_{n} (t) = t - U_{n} (t) = t - \max_i \{S_{n,i}: D_{n,i} \leq t\},
\end{align}
which is the time difference between  the current time $t$ and the generation time $U_{n} (t)$ of the freshest packet currently available at destination $d_n$. Because $S_{n,i}\leq D_{n,i}$,  one can get  $\Delta_{n} (t)\geq 0$ for all flow $n$ and time $t$. Let $\bm{\Delta}(t)=(\Delta_{1} (t),\ldots,\Delta_{N} (t)) \in [0,\infty)^N$ be the age vector of the $N$ flows at time $t$.

%represents the staleness of the information available at node $d_n$.
%\begin{align}
%U_{n}(t) = \max\{S_{n,i}: D_{n,i} \leq t\}.
%\end{align}
%to denote the time-stamp of the freshest update packet received by  up to time $t$. 
%The \emph{age of information}, or simply the \emph{age}, of flow $n$ is defined as
%\begin{align}\label{eq_age}
%\Delta_{n} (t) = t - U_{n}(t),
%\end{align}
%which represents the staleness of the available information at  flow $n$. 

We introduce an \emph{age penalty function} $p(\bm{\Delta}) = p\circ \bm{\Delta}$  to represent the level of dissatisfaction for having aged information at the $N$ destinations, where $p: [0,\infty)^N\rightarrow \mathbb{R}$ can be any  \emph{non-decreasing} function of the $N$-dimensional age vector  $\bm{\Delta}$. Some  examples of the age penalty function are: %in $\mathcal{P}_{\text{Sch}}$ are:
\begin{itemize}
\item[1.] The \emph{average age} of the $N$ flows is
\begin{align}\label{eq_avgage}
p_{\text{avg}} (\bm{\Delta}) = \frac{1}{N}\sum_{n=1}^N \Delta_{n}. 
\end{align}

\item[2.] The \emph{maximum age} of the $N$ flows is
\begin{align}%\label{eq_maxage}
p_{\max} (\bm{\Delta}) = \max_{n=1,\ldots,N} \Delta_{n}.
\end{align}

\item[3.] The \emph{mean square age} of the $N$ flows is
\begin{align}%\label{eq_msage}
p_{\text{ms}} (\bm{\Delta}) = \frac{1}{N}\sum_{n=1}^N (\Delta_{n} )^2.
\end{align}

\item[4.] The \emph{$l$-norm of the age vector} of the $N$ flows is
\begin{align}%\label{eq_msage}
p_{\text{$l$-norm}} (\bm{\Delta}) = \left[\sum_{n=1}^N (\Delta_{n} )^l\right]^{\frac{1}{l}}, ~l\geq1.
\end{align}

%\item[5.] The \emph{proportional fair age utility} of the flows is
%\begin{align}%\label{eq_msage}
%\Delta_{\text{PF}} (t) = \sum_{n=1}^N \log[\Delta_{n} (t)+ \epsilon],\nonumber
%\end{align}
%where $\epsilon>0$ is any positive number.\footnote{We are }

\item[5.] The \emph{sum of per-flow age penalty functions} is
\begin{align}%\label{eq_msage}
p_{\text{sum-penalty}} (\bm{\Delta}) = \sum_{n=1}^N g(\Delta_{n}),
\end{align}
where $g: [0,\infty) \rightarrow \mathbb{R}$ is a \emph{non-decreasing} function. Practical applications of non-decreasing age functions can be found in \cite{SunNonlinear2019,yates2021AgeSurvey,shisher2021age,ShisherMobiHoc22,shisher2023learning}. 
% \cite{Suninfocom2016,AgeOfInfo2016}. 

%{\red revise later}

%For example, a stair-shape function $g_1(\Delta)=\lfloor a \Delta\rfloor$ with $a\geq 0$ can be used to characterize the dissatisfaction of data staleness when the information of interests is checked periodically, and an exponential function $g_2(\Delta) = e^{a \Delta}$ is appropriate for online learning and control applications where the desire for information refreshing grows quickly with respect to the age \cite{AgeOfInfo2016}. 
\end{itemize}

In this paper, we consider a class of \emph{symmetric} and \emph{non-decreasing} age penalty functions, i.e.,
\begin{align}%\label{eq_class}
\mathcal{P}_{\text{sym}}
\!=\!\{p: [0,\infty)^N\rightarrow \mathbb{R}  \text{ is symmetric and non-decreasing}\}.\nonumber
\end{align}
This is a fairly large class of age penalty functions, where the function $p$ can be discontinuous, non-convex, or non-separable.
It is easy to see 
\begin{align}
\{p_{\text{avg}},p_{\max}, p_{\text{ms}},p_{\text{$l$-norm}}, %\Delta_{\text{PF}}, 
p_{\text{sum-penalty}}\}\subset \mathcal{P}_{\text{sym}}.
\end{align}
%Notice that $p(\bm{\Delta})$ is a function of time $t$ and policy $\pi$. 

In this paper, we consider both continuous-time and discrete-time  status updating systems. In the continuous-time setting, time $t \in [0,\infty)$ can take any positive value and the packet transmission times are \emph{i.i.d.} continuous random variables. On the other hand, in the discrete-time setting, time is quantized into multiples of a fundamental time unit $T_s$, i.e., $t \in \{0, T_s, 2T_s, \ldots\}$, and each packet's transmission time is fixed and equal to $T_s$. Consequently, the variables $S_{n,i}, A_{n,i}, D_{n,i}, t, U_{n} (t), \Delta_{n} (t)$ are all multiples of $T_s$. In realistic discrete-time systems, service preemption is not allowed. 

%For convenience, we set $T_s = 1$ second, thereby making all the discrete-time variables integers. The results for other values of $T_s$ can be easily derived by time scaling.

Let ${\Delta}_{n,\pi}(t)$ denote the age of flow $n$ achieved by scheduling policy $\pi$ 
and $\bm{\Delta}_{\pi}(t) = (\Delta_{1,\pi} (t),\ldots,\Delta_{N,\pi} (t))$. 
In the continuous-time case, we assume that the initial age $\bm\Delta_{\pi}(0^-)$ at time $t=0^-$ remains the same for all scheduling policies $\pi\in\Pi$, where $t=0^-$ is the moment right before $t=0$. In the discrete-time case, we assume that the initial age $\bm\Delta_{\pi}(0)$ at time $t=0$ remains the same for all scheduling policies $\pi\in\Pi$. 

The results in this paper remain true even if the age penalty function $p_t$ varies over time $t$. For example, it is allowed  that $p_t = p_{\text{avg}}$ for $0\leq t \leq 100$ and $p_t = p_{\text{max}}$ for $100<t \leq 200$. 
In the continuous-time case, we use $\{p_t \circ \bm{\Delta}_\pi(t), t\in [0,\infty)\}$ to represent the age-penalty stochastic process formed by the \emph{time-dependent} penalty function $p_t$ of the age vector $\bm{\Delta}_{\pi}(t)$ under scheduling policy $\pi$. In the discrete-time case, the age-penalty stochastic process is denoted by $\{p_t \circ \bm{\Delta}_\pi(t), t=0,T_s,2 T_s,\ldots\}$. 

\section{Multi-flow Status Update Scheduling: \\ The Continuous-time Case}\label{sec_analysis}
In this section, we investigate multi-flow scheduling in continuous-time status updating systems. We first consider a system setting with multiple servers and exponential transmission times, where an age-optimal scheduling result is established. Next, we study a more general system setting with multiple servers and NBU transmission times. In the second setting, age optimality is inherently difficult to achieve and we present a near age-optimal scheduling result.

\subsection{Multiple Flows, Multiple Servers, Exponential Service Times}
%When there is a single flow, the scheduler needs to  decide which packet to serve the first. In this case, it is known that the  \emph{Last Generated First Served (LGFS)}  scheduling discipline can achieve the minimum age process in a stochastic ordering sense \cite{Bedewy2016,BedewyJournal2017,Bedewy2017,BedewyMultihop2017}. 
%
%When there are multiple flows, the scheduler needs to compare not only the packets from the same flow, but also the packets from different flows, which makes the scheduling problem more complicated. 

To address the multi-flow scheduling problem, we consider a flow selection discipline called \emph{Maximum Age First  (MAF)}  \cite{LiInfocom2015,IgorAllerton2016,HsuTWC2017}, in which 
\emph{the flow with the maximum age is served first, with ties broken arbitrarily}. 
%(ii) The second discipline is called \emph{Maximum Age Reduction (MAR) first}: If the $i$-th packet of flow $n$ is fresher than any packet available at the destination $d_n$, the age reduction brought by the delivery of this packet is $S_{n,i}-U_n(t)$. In the MAR discipline,\emph{the packet with the maximum  age reduction is served first, with ties broken arbitrarily}.

For multi-flow single-server systems, a scheduling policy is defined by combining the Preemptive, MAF, and LGFS service disciplines as follows:

\begin{definition} \emph{Preemptive, Maximum Age First, Last Generated First Served (P-MAF-LGFS) policy:} This is a work-conserving scheduling policy for multiple-server, continuous-time systems with synchronized packet generations and arrivals. It operates as follows:

\begin{itemize}
\item[1.] If the queue is not empty, a server is assigned to process the most recently generated packet from the flow with the maximum age, with ties broken arbitrarily. 

\item[2.] The next server is assigned to process the most recently generated packet from the flow with the second maximum age, with ties broken arbitrarily. 

\item[3.] This process continues until either (i) the most recently generated packet of every flow is under service or has been delivered, or (ii) all servers are busy. 

\item[4.] If the most recently generated packet of every flow is under service or has been delivered, the remaining servers can be arbitrarily assigned to send the remaining packets in the queue, until the queue becomes empty.

\item[5.] When fresher packets arrive, the scheduler can preempt the packets that are currently under service and assign the new packets to servers following Steps 1-4 above. The preempted packets are then returned to the queue, where they await their turn to be transmitted at a later time.
\end{itemize}

%each packet assigned to the server is the last generated packet from the flow with the maximum age, , with ties broken arbitrarily.
%The scheduler first picks the last generated packet from each flow and assign these packets to idle servers according to the MA discipline; if there exist idle servers after the first round, the scheduler picks the second last generated packet from each flow and assign these packets to idle servers according to the MA discipline; this procedure continues until all servers are busy or all packets are under service. 
\end{definition}

The following observation provides useful insights into the operations of the P-MAF-LGFS policy: Due to synchronized packet generations and arrivals,  when the most recently generated packet of flow $n$ is successfully delivered in the P-MAF-LGFS policy, flow $n$ must have the \emph{minimum} age among the $N$ flows. Conversely, if flow $n$ does not have the \emph{minimum} age among all the flows, its most recently generated packet must be undelivered. Hence, in the P-MAF-LGFS policy, the most recently generated packet from a flow that does not have the \emph{minimum} age  is always available to be scheduled. 

The above P-MAF-LGFS policy is suitable for use in both single-server and multiple-server systems. It extends the original single-server P-MAF-LGFS policy introduced in \cite{SunAoIWorkshop2018} to encompass the more general multi-server scenario. 

%whether the preempted packets are dropped or stored back to the queue does not affect the age performance of the P-MAF-LGFS policy
%In the special case that there is a single flow ($N=1$), the MAF-LGFS policy reduces to the LGFS policy studied in \cite{Bedewy2016,BedewyJournal2017,Bedewy2017,BedewyMultihop2017}. 
%Hence, in the MAF-LGFS policy, the  LGFS discipline is adopted with a higher priority than the MA discipline. 
%The packet service priorities of the MAF-LGFS policy are illustrated in Fig. \ref{fig_policies}. 
The age optimality of the P-MAF-LGFS policy is established in Theorem \ref{thm1} and Corollary \ref{coro1} below.

\begin{theorem}(Continuous-time, multiple flows, multiple servers, exponential transmission times with transmission errors)\label{thm1}
In continuous-time status updating systems, if (i) the transmission errors are \emph{i.i.d.} with an error probability $q \in [0, 1)$, (ii)  the packet generation and arrival times are {synchronized} across the $N$ flows, and (iii)  the packet transmission times are exponentially distributed and \emph{i.i.d.} across packets, then it holds that for all $\mathcal{I}$, all $p_t \in\mathcal{P}_{\text{sym}}$, and all $\pi\in\Pi$ 
\begin{align}\label{thm1eq1}
&[\{p_t \circ\bm{\Delta}_{\text{P-MAF-LGFS}}(t), t\in [0,\infty)\}\vert\mathcal{I}] \nonumber\\
\leq_{\text{st}} & [\{p_t \circ\bm{\Delta}_\pi(t), t\in [0,\infty)\}\vert\mathcal{I}],
\end{align}
or equivalently, for all $\mathcal{I}$, all $p_t \in\mathcal{P}_{\text{sym}}$, and all non-decreasing functional $\phi$
 \begin{align}\label{thm1eq2}
&\mathbb{E}\left[\phi (\{p_t \circ\bm{\Delta}_{\text{P-MAF-LGFS}}(t),t\in [0,\infty)\})\vert\mathcal{I}\right] \nonumber\\
= & \min_{\pi\in\Pi} \mathbb{E}\left[\phi (\{p_t \circ\bm{\Delta}_\pi(t),t\in [0,\infty)\})\vert\mathcal{I}\right],
\end{align}
provided that the expectations in \eqref{thm1eq2} exist.
%, where $\mathbb{V}$ is the set of Lebesgue measurable functions defined in \eqref{eq_functions}.
\end{theorem}
\begin{proof}
%We develop a  sample-path method to prove Theorem \ref{thm1}.
%In the P-MAF-LGFS policy, if a packet from flow $n^*$ is delivered to its destination, then flow $n^*$ must be the flow with the \emph{maximum} age before the packet is delivered. Because the packet generation and arrival times are synchronized across the flows, flow $n^*$ is also the flow with the \emph{minimum} age after the packet is delivered.\footnote{If packet generations or arrivals are asynchronized, then this property is not guaranteed to hold.} Theorem \ref{thm1} is proven by employing this property to construct a sample-path argument. 
%The details are provided in 

See Appendix \ref{app1}.  \end{proof}

According to Theorem \ref{thm1}, for any age penalty function in $\mathcal{P}_{\text{sym}}$, any number of flows  $N$, any number of servers $M$,  any synchronized packet generation and arrival times in $\mathcal{I}$, and regardless the presence of \emph{i.i.d.} transmission errors or not,
the P-MAF-LGFS policy minimizes the stochastic process $[\{p_t \circ\bm{\Delta}_\pi(t), t\in [0,\infty)\}|\mathcal{I}]$ among all causal policies  in terms of stochastic ordering. Theorem \ref{thm1} is more general than \cite[Theorem 1]{SunAoIWorkshop2018}, as the latter was established for the special case of single-server systems without transmission errors.

By considering a mixture over the different realizations of $\mathcal{I}$, it can be readily deduced from Theorem \ref{thm1} that 
\begin{corollary}\label{coro1}
Under the conditions of Theorem \ref{thm1}, it holds that for all $p_t \in\mathcal{P}_{\text{sym}}$ and all $\pi\in\Pi$ 
\begin{align}\label{coro1eq1}
\{p_t \circ\bm{\Delta}_{\text{P-MAF-LGFS}}(t), t\in [0,\infty)\} \!\leq_{\text{st}} \!\{p_t \circ\bm{\Delta}_\pi(t), t\in [0,\infty)\}, 
\end{align}
or equivalently, for all $p_t \in\mathcal{P}_{\text{sym}}$ and all non-decreasing functional $\phi$
 \begin{align}\label{coro1eq2}
&\mathbb{E}\left[\phi (\{p_t \circ\bm{\Delta}_{\text{P-MAF-LGFS}}(t),t\in [0,\infty)\})\right] \nonumber\\
= & \min_{\pi\in\Pi} \mathbb{E}\left[\phi (\{p_t \circ\bm{\Delta}_\pi(t),t\in [0,\infty)\})\right],
\end{align}
provided that the expectations in \eqref{coro1eq2} exist.
%, where $\mathbb{V}$ is the set of Lebesgue measurable functions defined in \eqref{eq_functions}.
\end{corollary}

Corollary \ref{coro1} states that the P-MAF-LGFS policy minimizes the stochastic process $\{p_t \circ\bm{\Delta}_\pi(t), t\in [0,\infty)\}$ in a stochastic ordering sense, outperforming all other causal policies.

\subsubsection{Status Update Scheduling with Packet Replications}

As discussed in Section \ref{sec:queuemodel}, our study has been centered on a scenario where different servers are not allowed to simultaneously transmit packets from the same flow. In this context, we have demonstrated the age-optimality of the P-MAF-LGFS policy in Theorem \ref{thm1}. However, in situations where multiple servers can transmit packets from the same flow, and packet replication is permitted, it becomes possible to create multiple copies of the same packet and transmit these copies concurrently across multiple servers. The packet is considered delivered once any one of these copies is successfully delivered; at that point, the other copies are canceled to release the servers. If the packet service times follow an \emph{i.i.d.} exponential distribution with a service rate of $\mu$, the $N$ servers can be equivalently viewed as a single, faster server with exponential service times and a higher service rate of $N\mu$. Additionally, this fast server exhibits \emph{i.i.d.} transmission errors with an error probability $q$. Our study also addresses this scenario. 

\begin{definition} \emph{Preemptive, Maximum Age First, Last Generated First Served policy with packet Replications  (P-MAF-LGFS-R):} In this policy, the last generated packet from the flow with the maximum age is served the first among all packets of all flows, with ties broken arbitrarily. This packet is replicated into $N$ copies, which are transmitted concurrently over the $N$ servers. The packet is considered delivered once any one of these $N$ copies is successfully delivered; at that point, the other copies are canceled to release the servers.
\end{definition}

By applying Theorem \ref{thm1} to this particular scenario with a single, faster server, we derive the following corollary.

\begin{corollary}\label{corollary_new}
Under the conditions of  Theorem \ref{thm1}, if packet replication is allowed, then it holds that for all $\mathcal{I}$, all $p_t \in\mathcal{P}_{\text{sym}}$, and all $\pi\in\Pi$ 
\begin{align}\label{corollary_neweq1}
&[\{p_t \circ\bm{\Delta}_{\text{P-MAF-LGFS-R}}(t), t\in [0,\infty)\}\vert\mathcal{I}] \nonumber\\
\leq_{\text{st}} & [\{p_t \circ\bm{\Delta}_\pi(t), t\in [0,\infty)\}\vert\mathcal{I}],
\end{align}
or equivalently, for all $\mathcal{I}$, all $p_t \in\mathcal{P}_{\text{sym}}$, and all non-decreasing functional $\phi$
 \begin{align}\label{corollary_neweq2}
&\mathbb{E}\left[\phi (\{p_t \circ\bm{\Delta}_{\text{P-MAF-LGFS-R}}(t),t\in [0,\infty)\})\vert\mathcal{I}\right] \nonumber\\
= & \min_{\pi\in\Pi} \mathbb{E}\left[\phi (\{p_t \circ\bm{\Delta}_\pi(t),t\in [0,\infty)\})\vert\mathcal{I}\right],
\end{align}
provided that the expectations in \eqref{corollary_neweq2} exist.
%, where $\mathbb{V}$ is the set of Lebesgue measurable functions defined in \eqref{eq_functions}.
\end{corollary}

\ignore{

The following theorem establishes the age-optimality of the P-MAF-LGFS policy in the presence of \emph{i.i.d.} transmission errors.

\begin{theorem}(Continuous-time, transmission errors, multiple flows, single server, exponential transmission times)\label{coro2}
If there are \emph{i.i.d.} transmission errors with an error probability $q\in (0,1)$ and the conditions (ii)-(iv) of Theorem \ref{thm1} holds, then the results of Theorem \ref{thm1} and Corollary \ref{coro1} remain true.  
\end{theorem}
\begin{proof}
In order to prove Theorem \ref{coro2}, we modify the sample-path argument in the proof of Theorem \ref{coro1} by adopting a new coupling lemma that can handle transmission errors. See Appendix \ref{appcoro2} for the details. 
\end{proof}
}

\subsection{Multiple Flows, Multiple Servers, NBU Service Times}
Next, we consider a more general system setting with multiple servers and a class of New-Better-than-Used (NBU)  transmission time distributions that include exponential distribution as a special case. 
%The next question we proceed to answer is whether for an important class of distributions that are more general than exponential, age-optimality or near age-optimality can be achieved. 

%NBU distributions are defined as follows.
\begin{definition}  \emph{New-Better-than-Used Distributions:} Consider a non-negative random variable $X$ with complementary cumulative distribution function (CCDF) $\bar{F}(x)=\Pr[X>x]$. Then, $X$ is said to be \emph{New-Better-than-Used (NBU)} if for all $t,\tau\geq0$
\begin{equation}\label{NBU_Inequality}
\bar{F}(\tau +t)\leq \bar{F}(\tau)\bar{F}(t).
\end{equation} 
Examples of NBU distributions include deterministic distribution, exponential distribution, shifted exponential distribution, geometric distribution, gamma distribution, and negative binomial distribution. 
\end{definition}

In the scheduling literature, optimal  scheduling results were successfully established for delay minimization in single-server queueing systems, e.g., \cite{Schrage68,Jackson55}, but it can become inherently difficult in the multi-server cases. In particular, minimizing the average delay in deterministic scheduling problems with more than one servers is NP-hard  \cite{Leonardi:1997}. Similarly, delay-optimal stochastic scheduling in multi-class, multi-server queueing systems is deemed to be quite difficult \cite{Weiss:1992,Weiss:1995,Dacre1999}. The key challenge in multi-class, multi-server scheduling is that \emph{one cannot combine the capacities of all the servers to jointly process the most important packet}. Due to the same reason, age-optimal scheduling in multi-flow, multi-server systems is quite challenging. In the sequel, we consider a  relaxed goal to seek for \emph{near} age-optimal scheduling of multiple information flows, where our proposed scheduling policy is shown to be within a small additive gap from the optimum age performance.   

%We first construct a lower bound of the age $\age_{n}(t)$: 

\begin{figure}
\centering 
\includegraphics[width=0.3\textwidth]{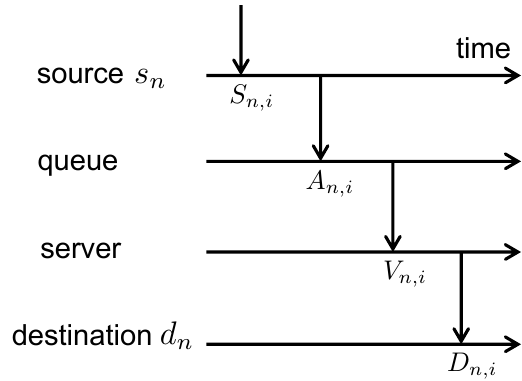} \caption{An illustration of $S_{n,i}$, $A_{n,i}$, $V_{n,i}$, and $D_{n,i}$.}
% work--efficiency ordering holds for any priorities of the jobs.
\label{fig_times1} 
\vspace{-5mm}
\end{figure}

%Notice that age $\age_{n}(t)$ in \eqref{eq_age} is determined by the packets that have been delivered to the destination $d_n$ by time $t$. 
To establish near age optimality, we introduce another age metric named \emph{age of served information}, denoted as $\Xi_{n} (t)$, which is a lower bound for age of information $\age_{n}(t)$: 

Let $V_{n,i}$ be the time that the $i$-th packet of flow $n$ starts its service by a server, i.e., the service starting time of the $i$-th packet of flow $n$. It holds that $S_{n,i}\leq A_{n,i}\leq V_{n,i}\leq D_{n,i}$, as illustrated in Fig. \ref{fig_times1}.
\emph{Age of served information} for flow $n$ is defined as
\begin{align}\label{eq_age_served}
\Xi_{n} (t) = t - \max_i\{S_{n,i}: V_{n,i} \leq t\},
\end{align}
which is the time difference between the current time $t$ and the generation time of the freshest packet that has started service by time $t$. Let $\bm{\Xi}(t)=(\Xi_{1} (t),\ldots,\Xi_{N} (t))$ be the age of served information vector at time $t$. Age of served information $\Xi_{n} (t)$ reflects the staleness of the packets that has begun service, whereas $\age_{n}(t)$ represents the staleness of the packets that has been successfully delivered to their destination. As depicted in Fig. \ref{fig_times2}, it is evident that $\Xi_{n} (t)\leq \age_{n}(t)$. In non-preemptive policies, the discrepancy between $\Xi_{n} (t)$ and $\age_{n}(t)$ solely arises from the \emph{i.i.d.} packet transmission times. 
Consequently, the age of served information $\Xi_{n} (t)$ closely approximates the age $\age_{n}(t)$.

We propose a new flow selection discipline called \emph{Maximum Age of Served Information First (MASIF)}, in which 
\emph{the flow with the maximum Age of Served Information is served first, with ties broken arbitrarily}. Using this discipline, we define another scheduling policy:

\begin{definition} \emph{Non-Preemptive, Maximum Age of Served Information first, Last Generated First Served (NP-MASIF-LGFS) policy:} This is a non-preemptive, work-conserving scheduling policy for multi-server systems. It operates as follows:
\begin{itemize}
\item[1.] When the queue is not empty and there are idle servers, an idle server is assigned to process the most recently generated packet from the flow with the maximum age of served information, with ties broken arbitrarily. 

\item[2.] After a packet from flow $n$ is assigned to an idle server, the server transitions into a busy state and will complete the transmission of the current packet from flow $n$ before serving any other packet. The age of served information $\Xi_{n} (t)$ of flow $n$ decreases. As a result, flow $n$ may no longer retain the maximum age of served information, allowing the remaining idle servers to be allocated to process other flows. The next idle server is assigned to process the most recently generated packet from the flow with the maximum age of served information, with ties broken arbitrarily. 

\item[3.] This procedure continues until either all servers are busy or the queue becomes empty. 

\end{itemize}

\end{definition}

%{\blue In some previous studies, e.g., \cite{IgorAllerton2016,HsuTWC2017,CostaCodreanuEphremides_TIT}, it was proposed to discard old packets and only store and send the freshest one. While this technique can reduce the age, in many applications such as social updates, news seeds, and stock trading, some old packets with earlier generation times are still quite useful and are needed to be sent to the destinations. 

Next, we will establish the near-age optimality of the NP-MASIF-LGFS policy. %Hence, the additional age reduction provided by discarding old packets in the NP-MASIF-LGFS policy is not large. 
The following theorem shows that the age of served information obtained by the NP-MASIF-LGFS policy serves as a lower bound (in terms of stochastic ordering) for the age of all other non-preemptive and causal policies. 

\begin{figure}
\centering 
\includegraphics[width=0.3\textwidth]{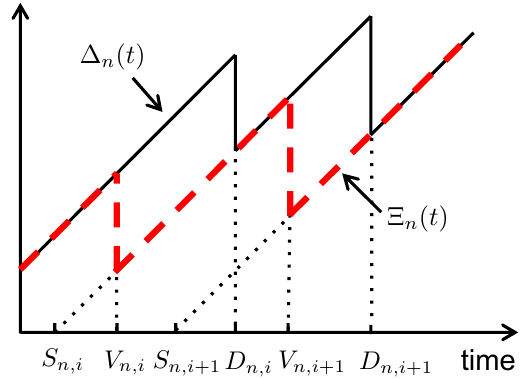} \caption{The age of served information $\Xi_{n} (t)$ as a lower bound of  age $\age_{n}(t)$.}
\vspace{-5mm}
% work--efficiency ordering holds for any priorities of the jobs.
\label{fig_times2} 
\end{figure} 

\begin{theorem}(Continuous-time, multiple flows, multiple servers, NBU transmission times with no errors) \label{thm3}
In continuous-time status updating systems, if  (i) there is no  transmission errors (i.e., $q=0$),  (ii) the packet generation and arrival times are {synchronized} across the $N$ flows,  and (iii) the packet transmission times are NBU and \emph{i.i.d.} across both servers and packets, then it holds that for all $\mathcal{I}$, all $p_t \in\mathcal{P}_{\text{sym}}$, and all $\pi\in\Pi_{np}$\footnote{Recall that $\Pi_{np}$ is the set of non-preemptive and causal scheduling policies.} 
\begin{align}\label{thm3eq1}
&[\{p_t \circ\bm{\Xi}_{\text{NP-MASIF-LGFS}}(t), t\in [0,\infty)\}\vert\mathcal{I}] \nonumber\\
\leq_{\text{st}} & [\{p_t \circ\bm{\Delta}_\pi(t), t\in [0,\infty)\}\vert\mathcal{I}],
\end{align}
or equivalently, for all $\mathcal{I}$, all $p_t \in\mathcal{P}_{\text{sym}}$, and all non-decreasing functional $\phi$
 \begin{align}\label{thm3eq2}
&\mathbb{E}\left[\phi (\{p_t \circ\bm{\Xi}_{\text{NP-MASIF-LGFS}}(t),t\in [0,\infty)\})\vert\mathcal{I}\right] \nonumber\\
\leq & \min_{\pi\in\Pi_{np}} \mathbb{E}\left[\phi (\{p_t \circ\bm{\Delta}_\pi(t),t\in [0,\infty)\})\vert\mathcal{I}\right] \nonumber\\
\leq & \mathbb{E}\left[\phi (\{p_t \circ\bm{\age}_{\text{NP-MASIF-LGFS}}(t),t\in [0,\infty)\})\vert\mathcal{I}\right],
\end{align}
provided that the expectations in \eqref{thm3eq2} exist.
\end{theorem}

\begin{proof}[Proof idea]
In the NP-MASIF-LGFS policy, if a packet from flow $n^*$ begins service, it implies that flow $n^*$ possesses the \emph{maximum} age of served information before the service starts. If the packet generation and arrival times are synchronized across the flows, flow $n^*$ also exhibits the \emph{minimum} age of served information after the service starts. The proof of Theorem \ref{thm3} relies on this property and a sample-path argument that is developed for NBU service time distributions. %{\blue We note that the sample-path method in \cite{sun2016delay,sun2017delay} is the key for addressing the challenge in multi-flow, multi-server scheduling.} 
\ifreport
See Appendix \ref{app2} for the details.
\else
See our technical report \cite{SunMultiFlow18} and \cite{sun2016delay,sun2017delay} for the details.
\fi
\end{proof}

%According to Theorem \ref{thm3}, the NP-MASIF-LGFS policy is near age-optimal in the sense of \eqref{thm3eq1} and \eqref{thm3eq2}. 

Considering the close approximation between the age of served information $\bm{\Xi}_{\text{NP-MASIF-LGFS}}(t)$ and the age of information $\bm{\age}_{\text{NP-MASIF-LGFS}}(t)$ in \eqref{thm3eq2}, we can conclude that the NP-MASIF-LGFS policy is near age-optimal.
Furthermore, in the case of the average age metric as defined in \eqref{eq_avgage} (i.e., $p_t = p_{\text{avg}} $ for all $t$), we can derive the following corollary: 
\begin{corollary}\label{coro4}
Under the conditions of Theorem \ref{thm3}, it holds that for all $\mathcal{I}$
\begin{align}\label{eq_gap}
\min_{\pi\in\Pi_{np}}\! [\bar{\age}_{ \pi}|\mathcal{I}] \!\leq\! [\bar{\age}_{\text{NP-MASIF-LGFS}}|\mathcal{I}]\!\leq\! \min_{\pi\in\Pi_{np}}\! [\bar{\age}_{ \pi}|\mathcal{I}] \!+\! \frac{1}{\mu},
\end{align}
where 
\begin{align}\label{eq_gap11}
[\bar{\age}_{ \pi}|\mathcal{I}] = \lim\sup_{T\rightarrow \infty} \frac{1}{T} \mathbb{E}\left[\int_0^T p_{\text{avg}} \circ\bm{\Delta}_\pi (t) dt \Bigg| \mathcal{I}\right]
\end{align} is the expected time-average of the average age of the $N$ flows, and $1/\mu$ is the mean packet transmission time.
\end{corollary}

\begin{proof}
The proof of Corollary \ref{coro4} is the same as that of Theorem 12 in \cite{BedewyJournal2017} and hence is omitted here. 
\end{proof}

By Corollary \ref{coro4}, the  average age of the NP-MASIF-LGFS policy is within an additive gap from the optimum, where the gap $1/\mu$ is invariant of the packet arrival and generation times $\mathcal{I}$, the number of flows $N$, and the number of servers $M$. 

Similar to Corollary \ref{coro1}, by taking a mixture over the different realizations of $\mathcal{I}$, one can remove the condition $\mathcal{I}$ from \eqref{thm3eq1}, \eqref{thm3eq2}, \eqref{eq_gap}, and \eqref{eq_gap11}.

The sampling-path argument utilized in the proof of Theorem \ref{thm3} can effectively handle multiple flows,  multiple servers, and \emph{i.i.d.} NBU transmission time distributions. This is achieved by establishing a coupling between the start time of packet transmissions in the NP-MASIF-LGFS policy and the completion time of packet transmissions in any other work-conserving policy from $\Pi_{np}$. However, extending this sampling-path argument to encompass the scenario of \emph{i.i.d.} transmission errors poses additional challenges that are currently difficult to overcome.

\section{Multi-flow Status Update Scheduling: \\ The Discrete-time Case}
In this section, we investigate age-optimal scheduling in discrete-time status updating systems, where the variables $S_{n,i}, A_{n,i}, D_{n,i}, t, U_{n} (t), \Delta_{n} (t)$ are all multiples of the period $T_s$, the transmission time of each packet is fixed as $T_s$, and the packet submissions are subject to \emph{i.i.d.} errors with an error probability $q\in[0,1)$. Service preemption is not allowed in discrete-time systems.

%In this discrete-time setting, the absence of packet arrivals during ongoing transmissions makes service preemption inconsequential and devoid of any advantages. As a result, we exclude service preemption from our design of the scheduling policy. 
For multiple-server, discrete-time systems, a scheduling policy is defined by combining the MAF and LGFS service disciplines as follows:

\begin{definition} \emph{Discrete Time, Maximum Age First, Last Generated First Served (DT-MAF-LGFS) policy:} This is a work-conserving scheduling policy for multiple-server, discrete-time systems with synchronized packet generations and arrivals. It operates as follows:

\begin{itemize}
\item[1.] When  time $t$ is a multiple of period $T_s$, if the queue is not empty, an idle server is assigned to process the most recently generated packet from the flow with the maximum age, with ties broken arbitrarily. 

\item[2.] The next  idle server is assigned to process the most recently generated packet from the flow with the second maximum age, with ties broken arbitrarily. 

\item[3.] This process continues until either (i) the most recently generated packet of each flow is under service or has been delivered, or (ii) all servers are busy. 

\item[4.] If the most recently generated packet of each flow is under service or has been delivered, and there are additional idle servers, then these servers can be arbitrarily assigned to send the remaining packets in the queue, until the queue becomes empty.
\end{itemize}
\end{definition}

One can observe that the DT-MAF-LGFS policy for discrete-time systems is similar to the P-MAF-LGFS policy designed for continuous-time systems. 

The age optimality of the DT-MAF-LGFS policy is obtained in the following theorem.

\begin{theorem}(Discrete-time, multiple flows, multiple servers, constant transmission times with transmission errors)\label{thm4}
In discrete-time status updating systems, if  (i) the transmission errors are \emph{i.i.d.} with an error probability $q\in [0,1)$, (ii)  the packet generation and arrival times are {synchronized} across the $N$ flows, and (iii) the packet transmission times are fixed as $T_s$, then it holds that for all $\mathcal{I}$, all $p_t \in\mathcal{P}_{\text{sym}}$, and all $\pi\in\Pi_{np}$ 
\begin{align}\label{thm4eq1}
&[\{p_t \circ\bm{\Delta}_{\text{DT-MAF-LGFS}}(t), t=0,T_s,2T_s,\ldots\}\vert\mathcal{I}] \nonumber\\
\leq_{\text{st}} & [\{p_t \circ\bm{\Delta}_\pi(t), t=0,T_s,2 T_s,\ldots\}\vert\mathcal{I}],
\end{align}
or equivalently, for all $\mathcal{I}$, all $p_t \in\mathcal{P}_{\text{sym}}$, and all non-decreasing functional $\phi$
 \begin{align}\label{thm4eq2}
&\mathbb{E}\left[\phi (\{p_t \circ\bm{\Delta}_{\text{DT-MAF-LGFS}}(t),t=0,T_s,2T_s,\ldots\})\vert\mathcal{I}\right] \nonumber\\
= & \min_{\pi\in\Pi_{np}} \mathbb{E}\left[\phi (\{p_t \circ\bm{\Delta}_\pi(t),t=0,T_s,2T_s,\ldots\})\vert\mathcal{I}\right],
\end{align}
provided that the expectations in \eqref{thm4eq2} exist.
%, where $\mathbb{V}$ is the set of Lebesgue measurable functions defined in \eqref{eq_functions}.
\end{theorem}

\begin{proof}See Appendix \ref{app_thm4}.  \end{proof}

\begin{figure}
\centering 
\includegraphics[width=0.45\textwidth]{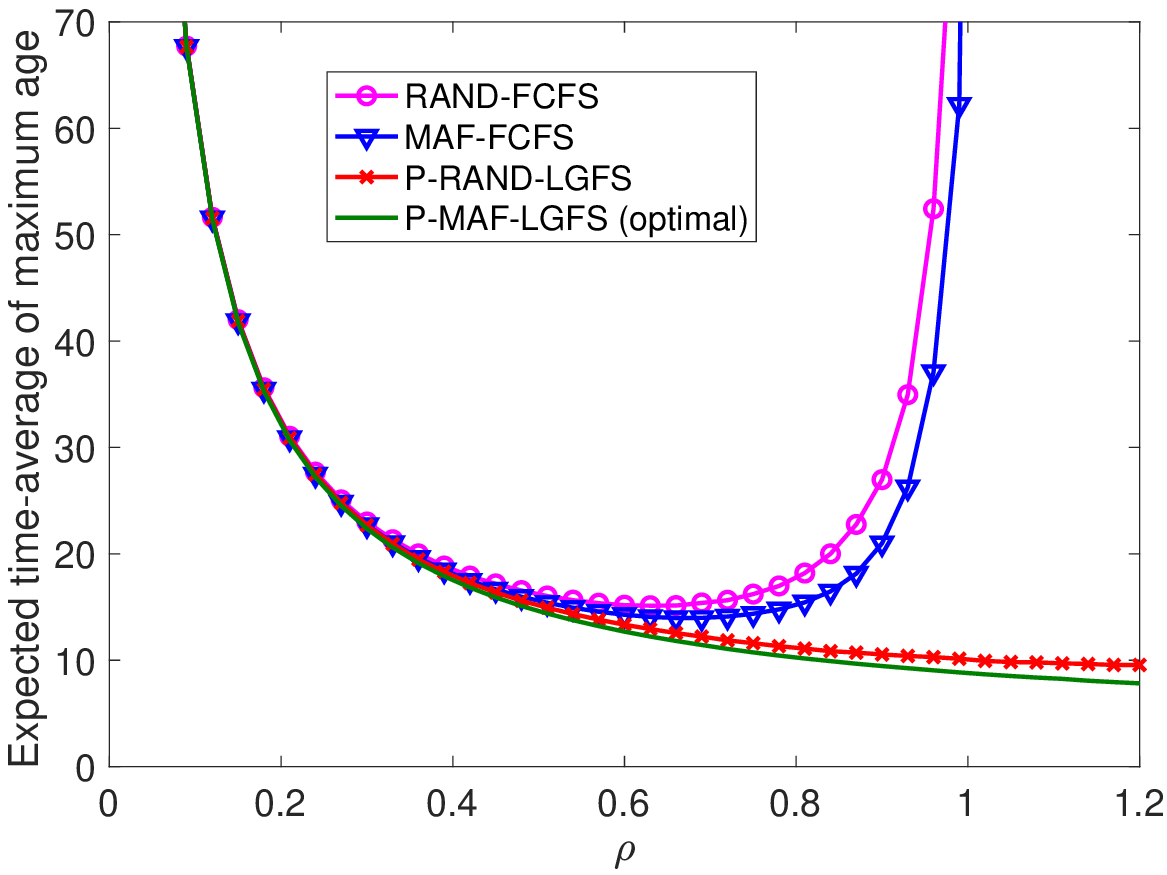} 

\caption{Expected time-average of the maximum age of 3 flows in a system with a single server and \emph{i.i.d.} exponential transmission times.}
\vspace{-5mm}
% work--efficiency ordering holds for any priorities of the jobs.
\label{fig_simulation1} 
\end{figure} 

According to  Theorem \ref{thm4}, the DT-MAF-LGFS policy minimizes the stochastic process $[\{p_t \circ\bm{\Delta}_\pi(t), t=0,T_s,2T_s,\ldots)\}|\mathcal{I}]$ in terms of stochastic ordering within discrete-time status updating systems. This optimality result holds for any age penalty function in $\mathcal{P}_{\text{sym}}$, any number of flows $N$, any number of servers $M$, any synchronized packet generation and arrival times in $\mathcal{I}$, and regardless the existence of \emph{i.i.d.} transmission errors. 

Theorem \ref{thm4} generalizes \cite[Theorem 1]{IgorAllerton2016}, %as the latter only holds in single-server systems with $p_t = p_{\text{avg}}$ for all $t$.
by allowing for multiple servers and a broader range of age penalty functions.
Similar to Corollary \ref{coro1}, one can remove the condition $\mathcal{I}$ from \eqref{thm4eq1} and \eqref{thm4eq2}.

\section{Numerical Results}
In this section, we evaluate the age performance of several multi-flow scheduling policies. These scheduling policies are defined by combining the flow selection disciplines $\{$MAF, MASIF, RAND$\}$ and the packet selection disciplines $\{$FCFS, LGFS$\}$, where  RAND represents randomly choosing a flow among the flows with un-served packets. The packet generation times $S_i$ follow a Poisson process with rate $\lambda$, and the time difference $(A_i-S_i)$ between packet generation and arrival is equal to either 0 or $4/\lambda$ with equal probability. The mean transmission time of each server is set as $\mathbb{E}[X]=1/\mu=1$. Hence, the traffic intensity is $\rho =  \lambda N/M$, where $N$ is the number of flows and $M$ is the number of servers.

\begin{figure}
\centering 
\includegraphics[width=0.45\textwidth]{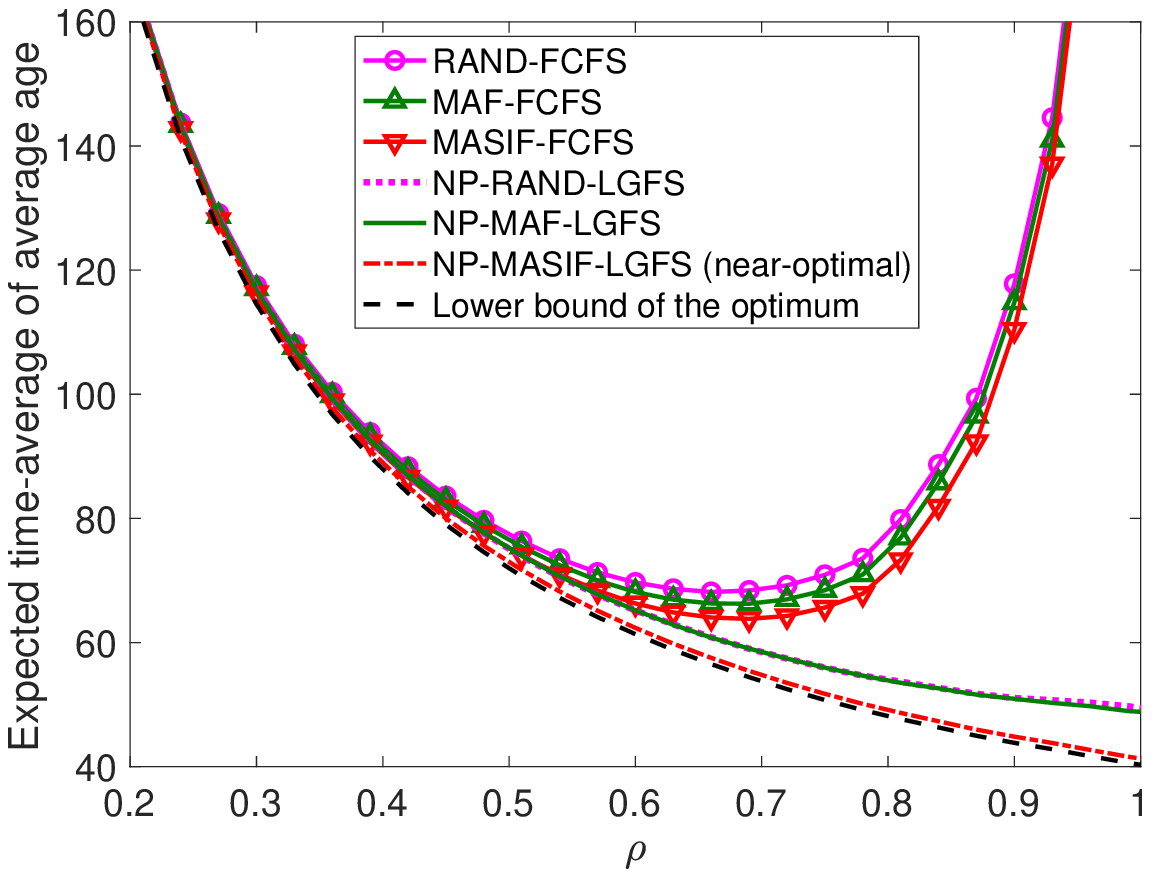} 
\caption{Expected time-average of the average age of 50 flows in a system with 3 servers and \emph{i.i.d.} NBU service times.}
% work--efficiency ordering holds for any priorities of the jobs.
\vspace{-5mm}

\label{fig_simulation2} 
\end{figure} 

Figure \ref{fig_simulation1} illustrates the expected time-average of the maximum age $p_{\max} (\bm\age(t))$ of 3 flows in a system with a single server and \emph{i.i.d.} exponential transmission times. One can see that the P-MAF-LGFS policy has the best age performance and its age is quite small even for $\rho>1$, in which case  the queue is actually unstable. On the other hand,  both the RAND and FCFS disciplines have much higher age. Note that there is no need for preemptions under the FCFS discipline.  Figure \ref{fig_simulation2} plots the expected time-average of the average age $p_{\text{avg}} (\bm\age(t))$ of 50 flows in a system with 3 servers and \emph{i.i.d.} NBU transmission times. In particular, the  transmission time $X$ follows the following shifted exponential distribution:
\begin{align}
\Pr[X>x] = \left\{\begin{array}{l l}1,&\text{if}~x<\frac{1}{3};\\
\exp[-\frac{3}{2}(x-\frac{1}{3})],&\text{if}~x\geq \frac{1}{3}.
\end{array}\right.
\end{align}
One can observe that the NP-MASIF-LGFS policy is better than the other policies, and is quite close to the age lower bound where the gap from the lower bound is no more than the mean transmission time $\mathbb{E}[X]=1$. {\blue One interesting observation is that the NP-MASIF-LGFS policy is better than the NP-MAF-LGFS policy for NBU transmission times. The reason behind this  is as follows: When multiple servers are idle,  the NP-MAF-LGFS policy will assign these servers to process multiple packets from the flow with the maximum age (say flow $n$). This reduces the age of flow $n$, but at a cost of postponing the service of the flows with the second and third maximum ages. On the other hand, in the NP-MASIF-LGFS policy, once a packet from the flow with the maximum age of served information  (say flow $m$) is assigned to a server, the age of served information of flow $m$ drops greatly, and the next server will be assigned to process the flow with the second maximum age of served information. 
As shown in \cite{sun2016delay,sun2017delay}, using multiple parallel servers to process different flows is beneficial for NBU service times. 

%The behavior of NP-MASIF-LGFS policy is similar to the maximum matching scheduling algorithms, e.g., \cite{Joo:2009,Ji2014} for time-slotted systems, where multiple servers are assigned to process different flows in each time-slot. One difference is that the NP-MASIF-LGFS policy can even operate in continuous-time systems, but the maximum matching scheduling algorithms cannot. 
%This phenomenon suggests that the MASIF discipline deserves further investigation, which will be a research task of our future studies.
 }
%These numerical results are in accordance with our theoretical analysis in Section \ref{sec_analysis}. 

\section{Conclusion}\label{sec_conclusion}
We have proposed causal scheduling policies and developed a unifying sample-path approach to prove that these scheduling policies are (near) optimal for minimizing age of information in continuous-time and discrete-time status updating systems with multiple flows, multiple servers, and transmission errors. 

\section*{Acknowledgement}

We appreciate Elif Uysal's support throughout this endeavor. Additionally, we thank the anonymous reviewers for their valuable comments.

%A policy $P\in\Pi$ is said to be \emph{age-optimal in stochastic ordering} for minimizing the age metric process $\{p \circ\bm{\Delta}_\pi(t), t\in [0,\infty)\}$ within the policy space $\Pi$, if for all $\pi\in\Pi$
%\begin{align}\label{eq_optimal}
%\{p \circ\bm{\Delta}_P(t), t\in [0,\infty)\} \leq_{\text{st}} \{p \circ\bm{\Delta}_\pi(t), t\in [0,\infty)\},
%\end{align}
%or equivalently, for all non-decreasing functional $\phi:\mathbf{V}\mapsto\mathbb{R}$
%\begin{align}\label{eq_optimal1}
%\mathbb{E}\left[\phi (p \circ\bm{\Delta}_P)\right] =  \min_{\pi\in\Pi} \mathbb{E}\left[\phi (p \circ\bm{\Delta}_\pi)\right]
%\end{align}
%provided the expectations in \eqref{eq_optimal1} exist, 
%\input{sec_solution}
%\input{sec_examples}
\bibliographystyle{IEEEtran}
\bibliography{ref,ref1,sueh,trialout}
% !TEX root = ./heterogeneous_servers.tex
\appendices
%\appendix

\section{Proof of Theorem \ref{thm1}}\label{app1}
Let the age vector $\bm\Delta_{\pi}(t)$ represent the \emph{system state} of policy $\pi$ at time $t$ and $\{\bm\Delta_{\pi}(t),t\in [0,\infty)\}$ be the \emph{state process} of policy $\pi$. For notational simplicity, let policy $P$ represent the P-MAF-LGFS policy, which is a work-conserving policy. We first establish two lemmas that are useful to prove Theorem \ref{thm1}. 
Using
the memoryless property of exponential distribution, we can
obtain the following coupling lemma:
\begin{lemma}\emph{(Coupling Lemma)}\label{coupling}
In continuous-time status up- dating systems, consider policy $P$ and any {work-conserving} policy $\pi\in \Pi$. For any given $\mathcal{I}$, if (i) the transmission errors are \emph{i.i.d.} with an error probability $q\in [0,1)$ and (ii)  the packet transmission times are exponentially distributed and \emph{i.i.d.} across  packets,
 then there exist policy $P_1$ and  policy $\pi_1$ in the same probability space which satisfy the same scheduling disciplines with policy $P$ and policy $\pi$, respectively,  such that 
\begin{itemize}
\itemsep0em 
\item[1.] the state process $\{\bm\Delta_{P_1}(t),t\in [0,\infty)\}$ of policy $P_1$ has the same distribution as the state process $\{\bm\Delta_{P}(t),t\in [0,\infty)\}$ of policy $P$,
\item[2.] the state process $\{\bm\Delta_{\pi_1}(t),t\in [0,\infty)\}$ of policy $\pi_1$ has the same distribution as the state process $\{\bm\Delta_{\pi}(t),t\in [0,\infty)\}$  of policy $\pi$,
\item[3.] if a packet from the flow with age $\Delta_{[i],P_1}(t)$ is successfully delivered at time $t$ in policy $P_1$, then almost surely, a packet from the flow with age $\Delta_{[i],\pi_1}(t)$ is successfully delivered at time $t$ in policy $\pi_1$; and vice versa.
%whenever there exist unassigned tasks in the queue,
\end{itemize} 
\end{lemma}
\ifreport
\begin{proof}
Notice that (i) all policies have identical packet generation and arrival times $\mathcal{I}$, (ii) the packet transmission times are  \emph{i.i.d.} memoryless, and (iii) policy $P$ and policy $\pi$ are both work-conserving. In addition, the packet generation/arrival times $\mathcal I$, the packet transmission times, and the transmission failures are governed by three mutually independent stochastic processes, none of which are influenced by the scheduling policy. Because of these facts, service preemption does not affect the distribution of packet delivery times. Following the inductive construction argument used in the proof of Theorem 6.B.3 in \cite{StochasticOrderBook}, one can construct the packet transmission success and failure events one by one in policy $P_1$ and policy $\pi_1$ to prove this lemma. 
In particular, since the transmission errors are \emph{i.i.d.} and they are not influenced by the scheduling policy, it is feasible to couple the packet transmission success and failure events in policy $P_1$ and policy $\pi_1$ in such a way that a packet from the flow with age $\Delta_{[i],P_1}(t)$ is successfully delivered at time $t$ in policy $P_1$ if, and only if, a packet from the flow with age $\Delta_{[i],\pi_1}(t)$ is successfully delivered at time $t$ in policy $\pi_1$. The details are omitted. 
\end{proof}
\else
\begin{proof}
See our technical report \cite{SunMultiFlow18}.
\end{proof}
\fi

We will compare policy $P_1$ and policy $\pi_1$ on a sample path by using the following lemma: 

\begin{lemma} \emph{(Inductive Comparison)}\label{lem2}
Suppose that a packet is delivered at time $t$ in policy $P_1$ and a packet is delivered at the same time $t$ in policy $\pi_1$. The system state  of policy $P_1$ is $\bm\Delta_{P_1}$ before the packet delivery, which becomes $\bm\Delta_{P_1}'$ after the packet delivery. The system state  of policy $\pi_1$ is $\bm\Delta_{\pi_1}$ before the packet delivery, which becomes $\bm\Delta_{\pi_1}'$ after the packet delivery. Under the conditions of Lemma \ref{coupling}, if (i) the packet generation and arrival times are {synchronized} across the $N$ flows and (ii) 
\begin{equation}\label{hyp1}
\Delta_{[i],P_1} \leq \Delta_{[i],\pi_1},~i=1,\ldots,N,
\end{equation}
then
\begin{equation}\label{law6}
\Delta_{[i],P_1}' \leq \Delta_{[i],\pi_1}',~i=1,\ldots,N.
\end{equation}  
\end{lemma}

\ifreport
\begin{proof}
For synchronized packet generations and arrivals, let $W(t) = \max_i\{S_i: A_i \leq t\}$ 
be the generation time of the freshest packet of each flow that has arrived at the queue by time $t$. At time $t$, because no packet that has arrived at the queue was generated later than $W(t)$, we can obtain
\begin{align}%\label{eq_proof_1}
%\Delta_{[i],P_1}'\geq t -W(t),~i=1,\ldots,N,\\
\Delta_{[i],\pi_1}' \geq t  -W(t),~i=1,\ldots,N.\label{eq_proof_2}
\end{align} 

Because (i) policy $P_1$ follows the same scheduling discipline with the P-MAF-LGFS policy and (ii) the packet generation and arrival times are {synchronized} across the $N$ flows, the delivered packet at time $t$ in policy $P_1$ must be the freshest packet generated at time $W(t)$. Hence, in policy $P_1$, the flow associated with the delivered packet must have the minimum age after the delivery, given by
\begin{align}
%\Delta_{[i],P_1}(t+T_s) &= \Delta_{[j_1],P_1},~i=1,\ldots,N-l,\\
\Delta_{[N],P_1}' &= t - W(t). \label{eq_proof_4}
\end{align}
Combining \eqref{eq_proof_2} and \eqref{eq_proof_4}, yields
\begin{align}
\Delta_{[N],P_1}' = t - W(t) \leq \Delta_{[N],\pi_1}'. \label{eq_proof_5}
\end{align}

Moreover, suppose that the packet delivered at time $t$ in policy $P_1$ is from the flow with age value $\Delta_{[j],P_1}$ before the packet delivery. This indicates 
\begin{align}
\Delta_{[i],P_1}' &= \Delta_{[i],P_1},~i=1,2,\ldots, j -1, \label{eq_proof_6}\\
\Delta_{[i],P_1}'& = \Delta_{[i+1],P_1},~i=j,2,\ldots, N-1. \label{eq_proof_7}
\end{align}
According to Lemma \ref{coupling}, the packet delivered at time $t$ in policy $\pi_1$ is from the flow with age value $\Delta_{[j],\pi_1}$ before the packet delivery. Hence, 
\begin{align}
\Delta_{[i],\pi_1}' &= \Delta_{[i],\pi_1},~i=1,2,\ldots, j -1, \label{eq_proof_8}\\
\Delta_{[i],\pi_1}'& \geq \Delta_{[i+1],\pi_1},~i=j,2,\ldots, N-1. \label{eq_proof_9}
\end{align}
Combining \eqref{hyp1}, \eqref{eq_proof_6}, and \eqref{eq_proof_8}, yields
\begin{align}
\Delta_{[i],P_1}' = \Delta_{[i],P_1} \leq \Delta_{[i],\pi_1} = \Delta_{[i],\pi_1}',~i=1,2,\ldots, j -1.\label{eq_proof_10}
\end{align}
Moreover, combining \eqref{hyp1}, \eqref{eq_proof_7}, and \eqref{eq_proof_9}, yields
\begin{align}
\Delta_{[i],P_1}' = \Delta_{[i+1],P_1} \leq \Delta_{[i+1],\pi_1} \leq \Delta_{[i],\pi_1}',\nonumber\\~i=j,2,\ldots, N-1.\label{eq_proof_11}
\end{align}
Finally, \eqref{law6} follows from \eqref{eq_proof_5}, \eqref{eq_proof_10}, and \eqref{eq_proof_11}. This completes the proof. 
\end{proof}
\else
\begin{proof}
See our technical report \cite{SunMultiFlow18}.
\end{proof}
\fi

%\begin{lemma}
%Consider two $N$-dimensional vectors $\bm{x}$ and $\bm{y}$. If $\bm{x}\leq\bm{y}$, then $x_{[i]} \leq y_{[i]}$ for all $i=1,\ldots,N$.
%\end{lemma}
%\begin{proof}
%For each $i=1,\ldots,N$, there exist $i$ elements $x_{[1]}, \ldots, x_{[i]}$ in $\bm{x}$ which are no smaller than $x_{[i]}
%$. This, together with $\bm{x}\leq\bm{y}$, tells us that at least $i$ elements in $\bm{y}$ are no smaller than $x_{[i]}
%$. Because $y_{[i]}$ is the $i$-th largest element in $\bm{y}$, $x_{[i]} \leq y_{[i]}$. This completes the proof.
%\end{proof}
Now we are ready to prove Theorem \ref{thm1}.
\begin{proof}[Proof of Theorem \ref{thm1}]
%See Appendix \ref{app1}.
Consider any work-conserving policy $\pi\in\Pi$. By Lemma \ref{coupling}, there exist policy $P_1$ and policy $\pi_1$
satisfying the same scheduling disciplines with policy $P$ and policy $\pi$, respectively, and the packet delivery times in policy $P_1$ and policy $\pi_1$ are synchronized almost surely.

For any given sample path of policy $P_1$ and policy $\pi_1$, $\bm\Delta_{P_1}(0^-) = \bm\Delta_{\pi_1}(0^-)$ at time $t=0^-$. We consider two cases:

\emph{Case 1:} When there is no packet delivery, the age of each flow grows linearly with a slope 1. 

\emph{Case 2:} When a packet is successfully delivered, the evolution of the system state is governed by  Lemma \ref{lem2}. 

By induction over time, we obtain
\begin{align}\label{eq_thm1_proof1}
\Delta_{[i],P_1} (t) \leq \Delta_{[i],\pi_1} (t),~i=1,\ldots,N,~t\geq 0.
\end{align}

For any symmetric and non-decreasing   function $p_t$, it holds from \eqref{eq_thm1_proof1} that for all sample paths and all $t\geq 0$
\begin{align}\label{eq_thm1_proof3}
&p_t\circ \bm \Delta_{P_1}(t) \nonumber\\
=& p_t(\Delta_{1,P_1} (t), \ldots, \Delta_{N,P_1} (t))\nonumber\\
=& p_t (\Delta_{[1],P_1} (t), \ldots, \Delta_{[N],P_1} (t))\nonumber\\
\leq & p_t (\Delta_{[1],\pi_1} (t), \ldots, \Delta_{[N],\pi_1} (t))\nonumber\\
=& p_t (\Delta_{1,\pi_1} (t), \ldots, \Delta_{N,\pi_1} (t))\nonumber\\
=& p_t\circ \bm \Delta_{\pi_1}(t).
\end{align}
By Lemma \ref{coupling}, the state process $\{\bm\Delta_{P_1}(t),t\in [0,\infty)\}$ of policy $P_1$ has the same distribution with the state process $\{\bm\Delta_{P}(t),t\in [0,\infty)\}$ of policy $P$;
the state process $\{\bm\Delta_{\pi_1}(t),t\in [0,\infty)\}$ of policy $\pi_1$ has the same distribution with the state process $\{\bm\Delta_{\pi}(t),t\in [0,\infty)\}$  of policy $\pi$. Hence, $\{p_t\circ\bm\Delta_{P_1}(t),t\in [0,\infty)\}$ has the same distribution with $\{p_t\circ\bm\Delta_{P}(t),t\in [0,\infty)\}$; $\{p_t\circ\bm\Delta_{\pi_1}(t),t\in [0,\infty)\}$ has the same distribution with $\{p_t\circ\bm\Delta_{\pi}(t),t\in [0,\infty)\}$. By substituting this and \eqref{eq_thm1_proof3} into Theorem 6.B.30 of \cite{StochasticOrderBook}, we can obtain that \eqref{thm1eq1} holds for all work-conserving policy $\pi\in\Pi$. 

For non-work-conserving policies $\pi$, because the service times are exponentially distributed (i.e., memoryless) and \emph{i.i.d.} across servers and time, server idling only postpones the delivery times of the packets. One can construct a coupling to show that for any non-work-conserving policy $\pi$, there exists a work-conserving policy $\pi'$ whose age process is smaller than that of policy $\pi$ in stochastic ordering; the details are omitted. 
%the age of 
%Therefore, the age under non-work-conserving policies will be greater. 
As a result, \eqref{thm1eq1} holds for all policies $\pi\in\Pi$.

Finally, the equivalence between \eqref{thm1eq1} and \eqref{thm1eq2} follows from \eqref{eq_order}. This completes the proof.
\end{proof}

\ifreport

\section{Proof of Theorem \ref{thm3}}\label{app2}

Let $(\bm\Delta_{\pi}(t),\bm\Xi_{\pi}(t))$ represent  the \emph{system state} of policy $\pi$ at time $t$ and $\{(\bm\Delta_{\pi}(t),\bm\Xi_{\pi}(t)),t\in [0,\infty)\}$ be the \emph{state process} of policy $\pi$. For notational simplicity, let policy $P$ represent the NP-MASIF-LGFS policy, which is a non-preemptive, work-conserving policy. 

For single-server  systems, the following \emph{work conservation law} 
plays an important role in the scheduling literature (see, e.g., \cite{Leonard_Kleinrock_book,Jose2010,Gittins:11}):
At any time, the expected total amount of time for completing the packets in the queue is invariant across different work-conserving policies. However, the work conservation law does not hold in multi-server  systems: It often happens that some servers are busy while the rest servers are idle, which leads to inefficient utilization of the idle servers and sub-optimal scheduling performance. 
In the sequel, we use a \emph{weak work-efficiency ordering} \cite{sun2016delay,sun2017delay} to compare different non-preemptive policies for multi-server systems. 

% proof is motivated by the sample-path  method developed in \cite{sun2016delay,sun2017delay} for near delay-optimal scheduling  in multi-server  systems.
%\footnote{Two work-efficiency orderings were used in \cite{sun2016delay,sun2017delay} to study (near) delay-optimal online scheduling in multi-server queueing systems.}

\begin{definition} \label{def_order} \emph{Weak Work-efficiency Ordering \cite{sun2016delay,sun2017delay}:}
For any given $\mathcal{I}$ and a sample path realization of two non-preemptive policies $\pi_1,\pi_2\in\Pi_{np}$, policy $\pi_1$ is said to be \emph{weakly more work-efficient} than policy $\pi_2$, if the following assertion is true:
%there is a one-to-one correspondence between the packets executed in policy $P$ and policy $\pi$ such that, if
{For each packet $j$ executed in policy $\pi_2$, if
\begin{itemize}
\item[1.] in policy $\pi_2$, a packet $j$ starts service at time $\tau$ and completes service at time $\nu$ ($\tau\leq \nu$), 
\item[2.] in policy $\pi_1$, the queue is not empty during $[\tau,\nu]$, 
\end{itemize}
then  in policy $\pi_1$, there always exists one corresponding packet $j'$ that starts service during $[\tau,\nu]$. It is worth noting that the weak work-efficiency ordering does not require to specify which servers are used to process packets $j$ and $j'$.} \end{definition}

\begin{figure}
\centering 
\includegraphics[width=0.3\textwidth]{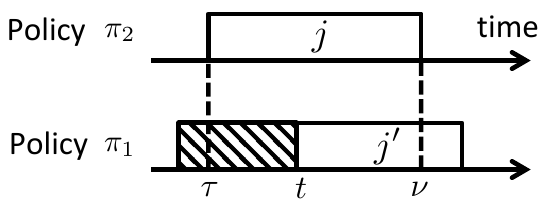} \caption{An illustration of the weak work-efficiency ordering, where the service duration of a packet is indicated by a rectangle, without specifying which servers are used to process the packets. 
%If a packet is replicated on multiple servers, then the service duration of this packet starts from the earliest time that one replica of this packet enters a server, until one replica of this packet is completed. 
Suppose that policy $\pi_1$ is {weakly more work-efficient than} policy $\pi_2$. If (i) a packet $j$ starts service at time $\tau$ and completes service at time $\nu$ in policy $\pi_2$, and (ii) the queue is not empty during $[\tau,\nu]$ in policy $\pi_1$, then in policy $\pi_1$ there exists one corresponding packet $j'$ that starts service at some time $t$ during $[\tau,\nu]$. 
}
% work--efficiency ordering holds for any priorities of the jobs.
\label{Work_Efficiency_Ordering} 
\end{figure} 

An illustration of the weak work-efficiency ordering is provided in Fig. \ref{Work_Efficiency_Ordering}. The weak work-efficiency ordering formalizes the following useful intuition for comparing two non-preemptive policies $\pi_1$ and $\pi_2$: \emph{If one packet $j$ is delivered at time $\nu$ in policy $\pi_2$, then  there exists one corresponding packet $j'$ that has started its transmission shortly before time $\nu$  in policy $\pi_1$, as long as the queue is not empty.} %Suppose that policy $\pi_1$ is {weakly more work-efficient than} policy $\pi_2$. If a packet $j$ starts service at time $\tau$ and completes service at time $\nu$ in policy $\pi_2$, then in policy $\pi_1$ either (i) there exists one corresponding packet $j'$ that starts service during $[\tau,\nu]$ or (ii) the queue is empty during $[\tau,\nu]$. 
The weak work-efficiency ordering  was originally introduced for near-optimal delay minimization in multi-server systems \cite{sun2016delay,sun2017delay}. In this paper, we use it for near-optimal age minimization in multi-server systems. 

%This is the key feature that enables us to establish a tight age lower bound. 

The following coupling lemma was established in \cite{sun2017delay} by using the property of NBU distributions:
\begin{lemma}\emph{(Coupling Lemma)} \cite[Lemma 2]{sun2017delay}\label{thm3lem_coupling}
In continuous-time status updating systems, consider two non-preemptive policies $P,\pi\in \Pi_{np}$. For any given $\mathcal{I}$, if (i) policy $P$ is work-conserving, and (ii) the packet service times are NBU, independent across the servers, and \emph{i.i.d.} across  the packets assigned to the same server, then there exist policy $P_1$ and  policy $\pi_1$ in the same probability space which satisfy the same scheduling disciplines with policy $P$ and policy $\pi$, respectively,  such that 
\begin{itemize}
\itemsep0em 
\item[1.] The state process $\{(\bm\Delta_{P_1}(t),\bm\Xi_{P_1}(t)),t\in [0,\infty)\}$ of policy $P_1$ has the same distribution as the state process $\{(\bm\Delta_{P}(t),\bm\Xi_{P}(t)),t\in [0,\infty)\}$ of policy $P$,
\item[2.] The state process $\{(\bm\Delta_{\pi_1}(t),\bm\Xi_{\pi_1}(t)),t\in [0,\infty)\}$ of policy $\pi_1$ has the same distribution as the state process $\{(\bm\Delta_{\pi}(t),\bm\Xi_{\pi}(t)),t\in [0,\infty)\}$  of policy $\pi$,
\item[3.] Policy $P_1$ is weakly more work-efficient than policy $\pi_1$ with probability one. 
%whenever there exist unassigned packets in the queue,
\end{itemize} 
\end{lemma}
The proof of Lemma \ref{thm3lem_coupling} is provided in \cite{sun2017delay}.

We will compare the age of service information of policy $P_1$ and the age of policy $\pi_1$ on a sample path by using the following lemma:

\begin{lemma} \emph{(Inductive Comparison)}\label{thm3lem2}
%Under the conditions of Lemma \ref{coupling}, 
Suppose that a packet starts service at time $t$ in policy $P_1$ and a packet completes service (i.e., delivered to the destination)  at the same time $t$ in policy $\pi_1$. The system state  of policy $P_1$ is $(\bm\Delta_{P_1},\bm\Xi_{P_1})$ before the service starts, which becomes $(\bm\Delta_{P_1}',\bm\Xi_{P_1}')$ after the service starts. The system state  of policy $\pi_1$ is $(\bm\Delta_{\pi_1},\bm\Xi_{\pi_1})$ before the service completes, which becomes $(\bm\Delta_{\pi_1}',\bm\Xi_{\pi_1}')$ after the service completes.
 If the packet generation and arrival times are {synchronized} across the $N$ flows and
\begin{equation}\label{thm3hyp1}
 \Xi_{[i],P_1} \leq \Delta_{[i],\pi_1},~i=1,\ldots,N,
\end{equation}
then
\begin{equation}\label{thm3law6}
\Xi_{[i],P_1}' \leq \Delta_{[i],\pi_1}',~i=1,\ldots,N.
\end{equation}  
\end{lemma}

\begin{proof}
For synchronized packet generations and arrivals,  let $W(t) = \max\{S_i: A_i \leq t\}$ 
be the generation time of the freshest packet of each flow that has arrived at the queue by time $t$. At time $t$, because no packet that has arrived at the queue was generated later than $W(t)$, we can obtain
\begin{align}%\label{eq_proof_1}
\Xi_{[i],P_1}' \geq t-W(t),~i=1,\ldots,N,\\
\Delta_{[i],\pi_1}' \geq t-W(t),~i=1,\ldots,N.\label{thm3eq_proof_2}
\end{align} 

Because policy $P_1$ follows the same scheduling discipline with the NP-MASIF-LGFS policy, each  packet starts service in policy $P_1$ must be from the flow with the maximum  age of served information $\Xi_{[1],P_1}$ (denoted as flow $n^*$), and the delivered packet must be the freshest packet that was generated at time $W(t)$. In other words, the age of served information of flow $n^*$ is reduced from the maximum age of served information $\Xi_{[1],P_1}$ to the minimum age of served information $\Xi_{[N],P_1}'=t-W(t)$, and the ages of served information of the other $(N-1)$ flows remain unchanged. Hence, 
\begin{align}\label{thm3eq_proof_3}
\Xi_{[i],P_1}' &= \Xi_{[i+1],P_1},~i=1,\ldots,N-1,\\
\Xi_{[N],P_1}' &= t - W(t). \label{thm3eq_proof_4}
\end{align}

In policy $\pi_1$, the delivered packet can be any packet from any flow. For all possible cases of policy $\pi_1$, it must hold that 
\begin{align}\label{thm3eq_proof_1}
\Delta_{[i],\pi_1}' \geq \Delta_{[i+1],\pi_1},~i=1,\ldots,N-1. 
\end{align}
By combining \eqref{thm3hyp1}, \eqref{thm3eq_proof_3}, and \eqref{thm3eq_proof_1}, we have
\begin{align}
\Delta_{[i],\pi_1}' \geq \Delta_{[i+1],\pi_1} \geq \Xi_{[i+1],P_1} = \Xi_{[i],P_1}',~i=1,\ldots,N-1.\nonumber
\end{align}
In addition, combining \eqref{thm3eq_proof_2} and \eqref{thm3eq_proof_4}, yields
\begin{align}
\Delta_{[N],\pi_1}' \geq  t-W(t) = \Xi_{[N],P_1}'.\nonumber
\end{align}
By this, \eqref{thm3law6} is proven.
\end{proof}
Now we are ready to prove Theorem \ref{thm3}.
\begin{proof}[Proof of Theorem \ref{thm3}]
Recall that policy $P$ is non-preemptive and work-conserving. 
 Consider any non-preemptive policy $\pi\in\Pi_{np}$.
By Lemma \ref{thm3lem_coupling}, there exist policy $P_1$ and policy $\pi_1$
satisfying the same scheduling disciplines with policy $P$ and policy $\pi$, respectively, and policy $P_1$ is weakly more work-efficient than policy $\pi_1$ with probability one.

Next, we construct another policy $\pi_1'$ in the same probability space with policy  $P_1$ and policy $\pi_1$: 
Because policy $P_1$ is weakly more work-efficient than policy $\pi_1$ with probability one, if
%After this modification, the following three properties are true:
\begin{itemize}
\item[1.] in policy $\pi_1$, a packet $j$ starts service at time $\tau$ and completes service at time $\nu$ ($\tau\leq \nu$),
\item[2.] in policy $P_1$, the queue is \emph{not empty} during $[\tau,\nu]$,
\end{itemize}
then  in policy $P_1$, there exists one corresponding packet $j'$ that starts service during $[\tau,\nu]$. Let  $t\in [\tau,\nu]$ be the service starting time of packet $j'$ in policy $P_1$, then  in policy $\pi_1'$, packet $j$ is constructed to start service at time $\tau$ and complete service at time $t$, as illustrated in Fig. \ref{Work_Efficiency_Ordering2}. On the other hand, if
%After this modification, the following three properties are true:
\begin{itemize}
\item[1.] in policy $\pi_1$, a packet $j$ starts service at time $\tau$ and completes service at time $\nu$ ($\tau\leq \nu$),
\item[2.] in policy $P_1$, the queue is \emph{empty} during $[\tau,\nu]$,
\end{itemize}
then in policy $\pi_1'$, packet $j$ is constructed to start service at time $\tau$ and complete service at time $\nu$. 
The initial age of policy $\pi_1'$ is chosen to be the same as that of other policies. Hence, $ \bm\Delta_{\pi_1'}(0^-) = \bm\Delta_{P_1}(0^-) = \bm\Delta_{\pi_1}(0^-)$.

The  policy $\pi_1'$ constructed above satisfies the following two useful properties: 
\begin{figure}
\centering 
\includegraphics[width=0.3\textwidth]{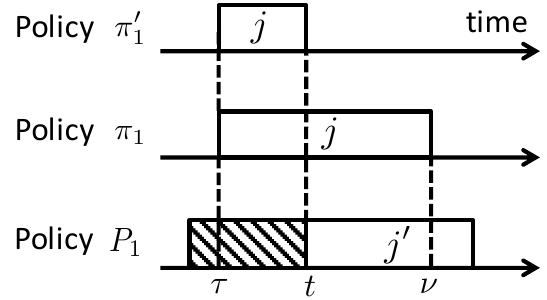} \caption{An illustration of the construction of policy $\pi'_1$, where the queue is \emph{not empty} during $[\tau,\nu]$ in policy $P_1$. The service completion time $t$ of packet $j$ in policy $\pi_1'$ is smaller than the service completion time $\nu$ of packet $j$ in policy $\pi$, and is equal to the service starting time $t$ of packet $j'$ in policy $P_1$. }
% work--efficiency ordering holds for any priorities of the jobs.
\label{Work_Efficiency_Ordering2} 
\end{figure}

\emph{Property (i):} The service completion time of each packet in policy $\pi_1'$ is equal to or earlier than that in policy $\pi$. 
Hence, 
\begin{align}\label{eq_thm3_proof2}
\bm\Delta_{\pi_1'}(t) \leq \bm\Delta_{\pi_1}(t), ~t\in[0,\infty)
\end{align}
holds with probably one. 

\emph{Property (ii):} If the queue is not empty at time $t$ in policy $P_1$ and a packet completes service  at time $t$ in policy $\pi_1'$, then a packet starts service at the same time $t$ in policy $P_1$.

Next, we use \emph{Property (ii)} to  show that, almost surely, 
\begin{align}\label{eq_thm3_proof1}
\Xi_{[i],P_1} (t) \leq \Delta_{[i],\pi_1'} (t),~i=1,\ldots,N,~t\geq 0.
\end{align}

At time $t= 0^-$, because $\bm\Xi_{P_1}(0^-) \leq \bm\Delta_{P_1}(0^-)$ and $\bm\Delta_{P_1}(0^-) = \bm\Delta_{\pi_1'}(0^-)$, we can obtain $\bm\Xi_{P_1}(0^-) \leq \bm\Delta_{\pi_1'}(0^-)$. This further implies that 
\begin{align}
\Xi_{[i],P_1} (0^-) \leq \Delta_{[i],\pi_1'} (0^-),~i=1,\ldots,N. 
\end{align}
For any time $t> 0$, there could be three cases:

\emph{Case 1:} if the queue is empty at time $t$ in policy $P_1$, then \eqref{eq_thm3_proof1} holds naturally at time $t$ because all packets have started services in policy $P_1$ (otherwise, the queue is not empty). 

\emph{Case 2:} if the queue is not empty at time $t$ in policy $P_1$ and a packet completes service at time $t$ in policy $\pi_1'$, according to \emph{Property (ii)}, a packet starts service at time $t$ in policy $P_1$. Hence, the evolution of the system state before and after time $t$ is governed by  Lemma \ref{thm3lem2}. 

\emph{Case 3:} if the queue is not empty at time $t$ in policy $P_1$ and no packet completes service at time $t$ in policy $\pi_1'$, there may exist some packet that starts service at time $t$ in policy $P_1$. Therefore, the age of each flow in policy $\pi_1'$  grows linearly with a slope 1 at time $t$; the age of served information of each flow in policy $P_1$ may either grow linearly with a slope 1 or drop to a lower value. By comparison, the age of served information of each flow in policy $P_1$ grows at a speed no faster than the age of each flow in policy $\pi_1'$.  

By induction over time and considering the above three cases, \eqref{eq_thm3_proof1} is proven.

Furthermore, for any symmetric and non-decreasing   function $p_t$, it holds from \eqref{eq_thm3_proof2} and \eqref{eq_thm3_proof1} that for all sample paths and all $t\geq 0$
\begin{align}\label{eq_thm3_proof3}
&p_t\circ \bm \Xi_{P_1}(t) \nonumber\\
=& p_t(\Xi_{1,P_1} (t), \ldots, \Xi_{N,P_1} (t))\nonumber\\
=& p_t (\Xi_{[1],P_1} (t), \ldots, \Xi_{[N],P_1} (t))\nonumber\\
\leq & p_t (\Delta_{[1],\pi_1'} (t), \ldots, \Delta_{[N],\pi_1'} (t))\nonumber\\
=& p_t (\Delta_{1,\pi_1'} (t), \ldots, \Delta_{N,\pi_1'} (t))\nonumber\\
=& p_t\circ \bm \Delta_{\pi_1'}(t)\nonumber\\
\leq & p_t\circ \bm \Delta_{\pi_1}(t).
\end{align}
By Lemma \ref{thm3lem_coupling}, the state process $\{(\bm\Delta_{P_1}(t),\bm\Xi_{P_1}(t)),t\in [0,\infty)\}$ of policy $P_1$ has the same distribution with the state process $\{(\bm\Delta_{P}(t),\bm\Xi_{P}(t)),t\in [0,\infty)\}$ of policy $P$;
the state process $\{(\bm\Delta_{\pi_1}(t),\bm\Xi_{\pi_1}(t)),t\in [0,\infty)\}$ of policy $\pi_1$ has the same distribution with the state process $\{(\bm\Delta_{\pi}(t),\bm\Xi_{\pi}(t)),t\in [0,\infty)\}$  of policy $\pi$. Hence, $\{p_t\circ\bm\Xi_{P_1}(t),t\in [0,\infty)\}$ 
has the same distribution with $\{p_t\circ\bm\Xi_{P}(t),t\in [0,\infty)\}$; $\{p_t\circ\bm\Delta_{\pi_1}(t),t\in [0,\infty)\}$ 
has the same distribution with $\{p_t\circ\bm\Delta_{\pi}(t),t\in [0,\infty)\}$. By substituting this and \eqref{eq_thm3_proof3} into Theorem 6.B.30 of \cite{StochasticOrderBook}, we can obtain that \eqref{thm3eq1} holds for all policy $\pi\in\Pi_{np}$. According to  \eqref{eq_order}, the first inequality in \eqref{thm3eq2} is equivalent to   \eqref{thm3eq1}. The second inequality in \eqref{thm3eq2} holds naturally. This completes the proof.
\end{proof}

\section{Proof of Theorem \ref{thm4}}\label{app_thm4}

Let the age vector $\bm\Delta_{\pi}(t) = (\Delta_{1,\pi} (t),\ldots,\Delta_{N,\pi} (t))$ represent the \emph{system state} of policy $\pi$ at time $t$ and $\{\bm\Delta_{\pi}(t),t=0,T_s,2T_s,\ldots\}$ be the \emph{state process} of policy $\pi$. Recall that $\Delta_{[i],\pi}(t)$ is the $i$-th largest component of the age vector $\bm\Delta_{\pi}(t)$. For notational simplicity, let policy $P$ represent the DT-MAF-LGFS policy, which is a non-preemptive, work-conserving policy. We first present two lemmas that are useful to prove Theorem \ref{thm4}.

\begin{lemma}\emph{(Coupling Lemma)}\label{coupling_4}
In discrete-time status updating systems, consider policy $P$ and any {non-preemptive, work-conserving} policy $\pi\in \Pi_{np}$. For any given $\mathcal{I}$, if (i) the transmission errors are \emph{i.i.d.} with an error probability $q\in [0,1)$ and (ii) the transmission time of each packet is equal to $T_s$,
 then there exist policy $P_1$ and policy $\pi_1$ in the same probability space which satisfy the same scheduling disciplines with policy $P$ and policy $\pi$, respectively,  such that 
\begin{itemize}
\itemsep0em 
\item[1.] the state process $\{\bm\Delta_{P_1}(t),t=0,T_s,2T_s,\ldots\}$ of policy $P_1$ has the same distribution as the state process $\{\bm\Delta_{P}(t),t=0,T_s,2T_s,\ldots\}$ of policy $P$,
\item[2.] the state process $\{\bm\Delta_{\pi_1}(t),t=0,T_s,2T_s,\ldots\}$ of policy $\pi_1$ has the same distribution as the state process $\{\bm\Delta_{\pi}(t),t=0,T_s,2T_s,\ldots\}$  of policy $\pi$,
\item[3.] if a packet from the flow with age $\Delta_{[i],P_1}(t)$ at time $t$ is successfully delivered  at time $(t + T_s)$ in policy $P_1$, then almost surely, a packet from the flow with age $\Delta_{[i],\pi_1}(t)$ at time $t$ is successfully delivered  at time $(t+ T_s)$ in policy $\pi_1$; and vice versa.
%whenever there exist unassigned tasks in the queue,
\end{itemize} 
\end{lemma}

\begin{proof}
By employing the inductive construction argument used in the proof of Theorem 6.B.3 in \cite{StochasticOrderBook}, one can construct the packet transmission success and failure events one by one in policy $P_1$ and policy $\pi_1$ to prove this lemma. 
In particular, since the transmission errors are \emph{i.i.d.} and they are not influenced by the scheduling policy, it is feasible to couple the packet transmission success and failure events in policy $P_1$ and policy $\pi_1$ in such a way that a packet from the flow with age $\Delta_{[i],P_1}(t)$ at time $t$ is successfully delivered at time $(t+ T_s)$ in policy $P_1$ if, and only if, a packet from the flow with age $\Delta_{[i],\pi_1}(t)$ at time $t$ is successfully delivered  at time $(t+ T_s)$ in policy $\pi_1$. \end{proof}

Notice that policy $P_1$ and policy $\pi_1$ are two distinct policies, so the flow with age $\Delta_{[i],P_1}(t)$ in policy $P_1$ and the flow with age $\Delta_{[i],\pi_1}(t)$ at time $t$ in policy $\pi_1$ are likely representing different flows. However, policy $P_1$ and policy $\pi_1$ are coupled in such a way  that  the packet deliveries for these two flows occur simultaneously at time $(t+ T_s)$.

We will compare policy $P_1$ and policy $\pi_1$ on a sample path by using the following lemma: 

\begin{lemma} \emph{(Inductive Comparison)}\label{thm4lem2}
Under the conditions of Lemma \ref{coupling_4}, if  (i) the packet generation and arrival times are {synchronized} across the $N$ flows and (ii) 
\begin{equation}\label{thm4lem2eq1}
\Delta_{[i],P_1}(t) \leq \Delta_{[i],\pi_1}(t),~i=1,\ldots,N,
\end{equation}
then
\begin{equation}\label{thm4lem2eq2}
\Delta_{[i],P_1}(t+T_s)\leq \Delta_{[i],\pi_1}(t+T_s),~i=1,\ldots,N.
\end{equation}  
\end{lemma}

\begin{proof}
For synchronized packet generations and arrivals, let $W(t) = \max_i\{S_i: A_i \leq t\}$ 
be the generation time of the freshest packet of each flow that has arrived at the queue by time $t$. Because (i) the packet transmission time is $T_s$ and (ii) no packet that has arrived at the queue by time $t$ was generated after time $W(t)$, we can obtain
\begin{align}%\label{eq_proof_1}
%\Delta_{[i],P_1}(t+T_s)\geq t + T_s-W(t),~i=1,\ldots,N,\\
\Delta_{[i],\pi_1}(t+T_s)\geq t + T_s -W(t),~i=1,\ldots,N.\label{thm4lem2eq3}
\end{align} 
Without loss of generality, suppose that there are $l$ transmission errors and $(N-l)$ successful packet deliveries at time $(t+T_s)$ in policy $P_1$. 
Because (i) policy $P_1$ follows the same scheduling discipline with the DT-MAF-LGFS policy and (ii) the packet generation and arrival times are {synchronized} across the $N$ flows, each delivered packet must be the freshest packet generated at time $W(t)$. Hence, the flows associated with these delivered packets must have the minimum age at time $(t+T_s)$, given by
\begin{align}
%\Delta_{[i],P_1}(t+T_s) &= \Delta_{[j_1],P_1},~i=1,\ldots,N-l,\\
\Delta_{[i],P_1}(t+T_s) &= t + T_s- W(t),~i=l+1,\ldots,N. \label{thm4lem2eq4}
\end{align}
Combining \eqref{thm4lem2eq3} and \eqref{thm4lem2eq4}, yields
\begin{align}
\Delta_{[i],P_1}(t+T_s) = t + T_s- W(t) \leq \Delta_{[i],\pi_1}(t+T_s),\nonumber\\ ~i=l+1,\ldots,N. \label{thm4lem2eq16}
\end{align}

Moreover, suppose that the transmission errors at time $(t+T_s)$ are from the flows with age values $(\Delta_{[j_1],P_1}(t),\Delta_{[j_2],P_1}(t),\ldots,\Delta_{[j_{l}],P_1}(t))$ at time $t$, which are sorted such that $j_1\geq j_2\geq \ldots \geq j_l$. Because $\Delta_{[i],P_1}(t)$ is the $i$-th largest component of the age vector $\bm\Delta_{P_1}(t)$, we have
\begin{align}
 \Delta_{[j_1],P_1}(t) \geq \Delta_{[j_2],P_1}(t)\geq\ldots\geq\Delta_{[j_{l}],P_1}(t). \label{thm4lem2eq8}
\end{align} 
If flow $n$ is one of the flows that encounter a transmission error at time $t+T_s$ in policy $P_1$, then
\begin{align}
\Delta_{n,P_1}(t+T_s) = \Delta_{n,P_1}(t) +T_s. \label{thm4lem2eq10}
\end{align}
%otherwise, it follows from \eqref{thm4lem2eq4} that flow $n$ has the minimum age
%\begin{align}
%\Delta_{n,P_1}(t+T_s) = t + T_s- W(t). \label{thm4lem2eq11}
%\end{align}
From \eqref{thm4lem2eq4}, \eqref{thm4lem2eq8}, and \eqref{thm4lem2eq10}, the components of vector $\bm\Delta_{P_1}(t+T_s)$
are $\Delta_{[j_1],P_1}(t)+T_s,\Delta_{[j_2],P_1}(t)+T_s,\ldots,\Delta_{[j_{l}],P_1}(t)+T_s$ and $(N-l)$ numbers with the same value $t + T_s- W(t)$. Hence,
\begin{align}
\Delta_{[i],P_1}(t+T_s) = \Delta_{[j_i],P_1}(t) +T_s, ~i=1,\ldots,l. \label{thm4lem2eq5}
\end{align}

According to Lemma \ref{coupling_4}, there are $l$ transmission errors at time $(t+T_s)$ in policy $\pi_1$, which are from the flows with age values $(\Delta_{[j_1],\pi_1}(t),\Delta_{[j_2],\pi_1}(t),\ldots,\Delta_{[j_{l}],\pi_1}(t))$ at time $t$. Because $j_1\geq j_2\geq \ldots \geq j_l$, we have
 \begin{align}
 \Delta_{[j_1],\pi_1}(t) \geq \Delta_{[j_2],\pi_1}(t)\geq\ldots\geq\Delta_{[j_{l}],\pi_1}(t). \label{thm4lem2eq12}
\end{align} 
If flow $n$ is one of the flows that encounter a transmission error at time $t+T_s$ in policy $\pi_1$, then
\begin{align}
\Delta_{n,\pi_1}(t+T_s) = \Delta_{n,\pi_1}(t) +T_s. \label{thm4lem2eq13}
\end{align}
From \eqref{thm4lem2eq12} and \eqref{thm4lem2eq13}, one can observe that $\Delta_{[j_1],\pi_1}(t)+T_s,\Delta_{[j_2],\pi_1}(t)+T_s,\ldots,\Delta_{[j_{l}],\pi_1}(t)+T_s$ are $l$ components of vector $\bm\Delta_{\pi_1}(t+T_s)$. Hence, 
\begin{align}
\Delta_{[i],\pi_1}(t+T_s) \geq \Delta_{[j_i],\pi_1}(t) +T_s, ~i=1,\ldots,l. \label{thm4lem2eq6}
\end{align}
Combining \eqref{thm4lem2eq1}, \eqref{thm4lem2eq5}, and \eqref{thm4lem2eq6}, yields
\begin{align}
&\Delta_{[i],P_1}(t+T_s) \nonumber\\
= &\Delta_{[j_i],P_1}(t) +T_s  \nonumber\\
\leq &\Delta_{[j_i],\pi_1}(t) +T_s \nonumber\\
\leq & \Delta_{[i],\pi_1}(t+T_s), ~i=1,\ldots,l. \label{thm4lem2eq7}
\end{align}
Finally, \eqref{thm4lem2eq2} follows from \eqref{thm4lem2eq16} and \eqref{thm4lem2eq7}. This completes the proof. 
\end{proof}

Now we  prove Theorem \ref{thm4}.
\begin{proof}[Proof of Theorem \ref{thm4}]
%See Appendix \ref{app1}.
Consider any non-preemptive, work-conserving policy $\pi\in\Pi_{np}$. By Lemma \ref{coupling_4}, there exist policy $P_1$ and policy $\pi_1$
satisfying the same scheduling disciplines with policy $P$ and policy $\pi$, respectively, such that if a packet from the flow with age $\Delta_{[i],P_1}(t)$ at time $t$ is successfully delivered  at time $(t + T_s)$ in policy $P_1$, then almost surely, a packet from the flow with age $\Delta_{[i],\pi_1}(t)$ at time $t$ is successfully delivered  at time $(t+ T_s)$ in policy $\pi_1$; and vice versa.

For any given sample path of policy $P_1$ and policy $\pi_1$, the initial system state is $\bm\Delta_{P_1}(0) = \bm\Delta_{\pi_1}(0)$ at time $t=0$. The evolution of the system state is governed by  Lemma \ref{thm4lem2}. By induction over time, we obtain
\begin{align}\label{eq_thm4_proof1}
\Delta_{[i],P_1} (t) \leq \Delta_{[i],\pi_1} (t),~i=1,\ldots,N,~t = 0, T_s, 2T_s,\ldots.
\end{align}
The rest of the proof is quite similar to that of Theorem \ref{thm1} and hence are omitted. 
\end{proof}

\fi

\end{document}